\documentclass[sigconf, nonacm]{acmart}

\newcommand\vldbdoi{XX.XX/XXX.XX}
\newcommand\vldbpages{XXX-XXX}
\newcommand\vldbvolume{14}
\newcommand\vldbissue{1}
\newcommand\vldbyear{2020}
\newcommand\vldbauthors{\authors}
\newcommand\vldbtitle{\shorttitle} 
\newcommand\vldbavailabilityurl{URL_TO_YOUR_ARTIFACTS}
\newcommand\vldbpagestyle{plain} 

\def\addlegendimage{\csname pgfplots@addlegendimage\endcsname}
\usepackage{multirow}
\usepackage{float}

\newcommand{\tocheck}[1]{\textcolor{black}{#1}}
\newcommand{\BF}{\textsf{LEARNT}}

\newcommand{\best}{\textsuperscript{$\star$}}
\newcommand{\second}{\textsuperscript{$\ddagger$}}

\newcommand{\revision}[1]{\textcolor{black}{#1}}

\usepackage{caption}
\usepackage{subcaption}
\usepackage{tikz}
\usepackage{pgfplots}
\usepackage{filecontents}
\usepackage{bbding}
\usepackage{threeparttable}
\usepackage{enumitem}
\usepackage{xcolor}
\usepackage{url}
\usepackage{multicol}
\usepackage{diagbox}
\usepackage[a-2b]{pdfx}

\newtheorem{obsv}{Observation}

\newtheorem{lemma}{Lemma}[section]
\usepackage{pifont}
\usepackage{dsfont}

\usepackage{lipsum}

\makeatletter
\newif\if@restonecol
\makeatother

\usepackage[linesnumbered,ruled,vlined]{algorithm2e}
\usepackage{algpseudocode}
\usepackage{amsmath}
\usepackage{booktabs}

\usepackage[colorinlistoftodos,prependcaption, textsize=small, textwidth=35]{todonotes}

\begin{document}
\title{\BF{}: A Practical Estimator for Cardinality of LIKE Queries
with Formal Accuracy Guarantees}

\author{Hai Lan}
\affiliation{%
  \institution{The University of Queensland}
  \city{Brisbane}
  \country{Australia}
  \postcode{43017-6221}
}
\email{h.lan@uq.edu.au}

\author{Zhifeng Bao}
\affiliation{%
  \institution{The University of Queensland}
  \city{Brisbane}
  \country{Australia}
}
\email{zhifeng.bao@uq.edu.au}

\author{Divesh Srivastava}
\affiliation{%
  \institution{AT\&T Chief Data Office}
  \city{Bedminster, NJ}
  \country{USA}
}
\email{divesh@research.att.com}

\author{Shixun Huang}
\affiliation{%
  \institution{University of Wollongong}
  \city{Wollongong}
  \country{Australia}
}
\email{shixunh@uow.edu.au}

\author{Yuwei Peng}
\affiliation{%
  \institution{Wuhan University}
  \city{Wuhan}
  \country{China}
}
\email{ywpeng@whu.edu.cn}

\author{Yang Yu}
\affiliation{%
  \institution{Wuhan University}
  \city{Wuhan}
  \country{China}
}
\email{yu_yang@whu.edu.cn}

\begin{abstract}
    \tocheck{We study the problem of cardinality estimation for \textsf{LIKE} queries on string data, focusing on the most common patterns in real workloads: prefix, suffix, and substring queries. 
    We propose \BF{}, a \textsf{\textbf{L}IKE} query \textbf{E}stimator with \textbf{A}ccuracy, \textbf{R}obustness, \textbf{N}egligible overhead, \textbf{T}unability, and \textbf{T}heoretical guarantees. \BF{} formulates estimation as a bucket-classification problem, and upon correct classification, it yields formal bounds on Q-error for the queries with non-empty answer. It employs a memory-efficient bucketed layered-filter architecture with Bloom filters and compact auxiliary tables, together with optimizations that exploit query skew to reduce storage. For the queries that have empty answer, \BF{} incorporates dedicated filter-based and prefix-walk strategies, providing probabilistic guarantees on correct identification. Furthermore, to support arbitrarily long query strings, we extend \BF{} with Markov modeling scheme that composes short-query statistics into estimates for longer queries. A theoretical framework guides parameter selection to minimize storage under accuracy and robustness constraints. Extensive experiments on four real-world datasets show that \BF{} consistently outperforms state-of-the-art methods such as CLIQUE and LPLM, achieving 1.3–1.7× lower mean Q-error, significantly lower tail errors, and up to 70× faster construction with comparable memory usage.}

\end{abstract}

\maketitle

\vspace{-2mm}
\pagestyle{\vldbpagestyle}
\begingroup\small\noindent\raggedright\textbf{PVLDB Reference Format:}\\
\vldbauthors. \vldbtitle. PVLDB, \vldbvolume(\vldbissue): \vldbpages, \vldbyear.\\
\href{https://doi.org/\vldbdoi}{doi:\vldbdoi}
\endgroup
\begingroup
\renewcommand\thefootnote{}\footnote{\noindent
This work is licensed under the Creative Commons BY-NC-ND 4.0 International License. Visit \url{https://creativecommons.org/licenses/by-nc-nd/4.0/} to view a copy of this license. For any use beyond those covered by this license, obtain permission by emailing \href{mailto:info@vldb.org}{info@vldb.org}. Copyright is held by the owner/author(s). Publication rights licensed to the VLDB Endowment. \\
\raggedright Proceedings of the VLDB Endowment, Vol. \vldbvolume, No. \vldbissue\ %
ISSN 2150-8097. \\
\href{https://doi.org/\vldbdoi}{doi:\vldbdoi} \\
}\addtocounter{footnote}{-1}\endgroup

\vspace{-2mm}
\ifdefempty{\vldbavailabilityurl}{}{
\vspace{.3cm}
\begingroup\small\noindent\raggedright\textbf{PVLDB Artifact Availability:}\\
The source code, data, and/or other artifacts have been made available at \url{https://github.com/DataAutonomyLab/ce4str}.
\endgroup
}


\section{Introduction}\label{sec:intro}
\tocheck{Cardinality estimation is a cornerstone of modern cost-based query optimizers~\cite{Howgood,OptSurvey}, and numerous approaches have been proposed for different data types~\cite{CE-Exp-VLDB22, CE-Exp-SIGMOD22, CE4HD, FACE, FLAT}. In this paper, we focus on cardinality estimation for \textsf{LIKE} queries. Motivated by the prevalence of prefix ($S\%$), suffix ($\%S$), and substring ($\%S\%$) patterns observed in real workloads (Table~\ref{tab:query-patterns-bench}), we concentrate on these three cases.} 

\revision{\noindent\textbf{Studied Problem.} Let a string $S$ be a sequence of characters $s_1 \ldots s_{|S|}$ drawn from a predefined vocabulary $\mathcal{V}$. A (prefix, suffix, or substring) \textsf{LIKE} query $Q$, $[\%]q_1 \ldots q_{|Q|}[\%]$, is a string where each $q_i \in \mathcal{V}$ and has at least one wildcard symbol $\%$. A string $S$ matches $Q$ if it satisfies the pattern of $Q$, denoted as ${I(S,Q)} = 1$; otherwise, ${I(S,Q)} = 0$. To this end, given a string dataset $\mathcal{S}$ and a \textsf{LIKE} query $Q$, the \textbf{cardinality estimation for $Q$ on $\mathcal{S}$} aims to estimate the number of strings in $\mathcal{S}$ that match $Q$, i.e., $\sum_{S \in \mathcal{S}}{I(S,Q)}$.}

\tocheck{Depending on whether $Q$ returns an empty result, we classify $Q$ as a \emph{non-empty-answer query} if it matches at least one string in the dataset, and as a \emph{empty-answer query} otherwise. As we will show shortly, empty-answer queries are often overlooked by previous work but very important to handle effectively.}

\begin{table}[t]
	\small
	\vspace{-2mm}
	\centering
	\caption{\textsf{LIKE} Query Pattern Frequencies in Real Benchmarks}
	\vspace{-4mm} 
	\resizebox{0.32\textwidth}{!}{%
		\begin{tabular}{|c|c|c|c|c|}
			\hline
            \textbf{Benchmark} & \textbf{$S\%$} & \textbf{$\%S$}& \textbf{$\%S\%$} & \textbf{Others} \\\hline
		JOB~\cite{JOB}&17&0&68&32\\\hline
            TPC-H~\cite{TPCH} &3&1&1&2\\\hline
            CEB~\cite{CEB}&0&0&4865&0\\\hline
            STACK~\cite{STACK}&8&1538&346&0\\\hline
            TPC-DS~\cite{TPCDS}&1&0&0&0\\\hline
            DSB~\cite{DSB}&1&0&0&0\\\hline
                		\end{tabular}%
	}
	\label{tab:query-patterns-bench}
	\vspace{-7mm}
\end{table}

\vspace{-2mm}
\subsection{Identified Research Gaps}
For a cardinality estimator to be \textbf{\emph{practical}} in a production database
system, accuracy alone is not sufficient. Since cardinality estimates are
invoked repeatedly during query optimization, an estimator must provide
reliable predictions across diverse query patterns and cardinality ranges,
avoid catastrophic errors that may mislead the optimizer, handle
empty-answer queries conservatively, and impose low construction, storage,
and inference overheads. These requirements are particularly important
for LIKE queries over string data, where query patterns, string
distributions, and workloads can vary significantly across applications.

Existing methods fall into two major categories: non-learning based approaches~\cite{LBS,MOKVI} and learning-based approaches~\cite{E2E,Astrid,LPLM,SSCard,CLIQUE} (see Sec.~\ref{sec:related_work} for details). Based on an in-depth survey, summarized in Table~\ref{tab:method-comparison}, we identify four key research gaps: 

\noindent\textit{G1: Insufficient Accuracy for Non-empty-answer Queries}. 
{{Learning-based methods have shown notable accuracy gains over traditional techniques such as PostgreSQL’s estimator~\cite{PG}. However, due to the complexity of string data distributions, these methods often struggle to generalize across datasets, query patterns, cardinality spectrum. For example, our experiments show that existing learning-based estimators yield high mean Q-errors on the DBLP-AN dataset (cf. Sec.\ref{sec:exp-acc}) and fail to handle low-cardinality and high-cardinality queries simultaneously (cf. Sec.\ref{sec:exp-basic-robust}). Such inconsistent behavior across the cardinality spectrum can lead to brittle or suboptimal query plans in cost-based optimizers.}}

\begin{table}[!ht]
	\small
	\vspace{-2mm}
	\centering
	\caption{Comparative Analysis of Representative Cardinality Estimation Methods for \tocheck{$S\%$, $\%S$, and $\%S\%$} Queries. Note: MO, LBS, and SSCard only support substring queries ($\%S\%$). 
    }
	\vspace{-4mm} 
	\resizebox{0.48\textwidth}{!}{%

            \begin{tabular}{|c|c|c|c|c|c|}
                  \hline
                  \textbf{Method}      &  \textbf{Acc. (G1)}      & \textbf{Err. Bound (G2)}  & \textbf{Empty-ans. (G3)} & \textbf{Prep. (G4)} & \textbf{Space (G4)} \\\hline
                  PostgreSQL~\cite{PG}  & $+$           & \ding{53}                   & $+$      & $+++$        & $+++$ \\\hline
                  MO~\cite{MOKVI}  & $++$           & \ding{53}                   & $+++$      & $++$        & $+$ \\\hline
                  LBS~\cite{LBS}  & $+$           & \ding{53}                   & $+++$      & $+$        & $+$ \\\hline\hline
                  E2E~\cite{E2E}         & $+$           & \ding{53}                   &      $+$   & $+$   & $+$ \\\hline
                  Astrid~\cite{Astrid}      & $++$      & \ding{53}                   &    $++$     & $+$      & $+++$ \\\hline
                  LPLM~\cite{LPLM}        & $+$           & \ding{53}                   & $++$      & $+$     & $+++$ \\\hline
                  CLIQUE~\cite{CLIQUE}      & $++$      & \ding{53}                  &   $+++$   & $+$      & $+++$ \\\hline
                  SSCard~\cite{SSCard}      & $++$      & \ding{53}                   &    $+++$     & $++$      & $++$ \\\hline\hline
                  \BF{}        & $+++$          & \checkmark                   & $+++$     & $+++$  & $+++$   \\
                  \hline
                \end{tabular}
	}
	\label{tab:method-comparison}
    \begin{flushleft}
\footnotesize
\textbf{Note}: (1) \emph{Acc.} refers to estimation accuracy. \emph{empty-ans.} refers to empty-answer queries. \emph{Prep.} indicates preprocessing time required before supporting queries. (2) More $+$ means higher accuracy or lower cost, \textbf{\checkmark}: supported, {\ding{53}}: not supported. 
\vspace{-4mm}
\end{flushleft}
\end{table}


\noindent\textit{G2: Lack of Robustness and Error Bound for Non-empty-answer Queries}. 
Existing estimators often fail to maintain accuracy across the full spectrum of cardinalities and varying query string lengths (cf. Sec.~\ref{sec:exp-basic-robust}). For instance, CLIQUE~\cite{CLIQUE} incurs large errors on high-cardinality queries, while LPLM~\cite{LPLM} underperforms in low-cardinality cases. More importantly, none of these methods offers formal error bounds -- individual estimates can be arbitrarily far from the truth.

\noindent\tocheck{\textit{G3: Unreliable Handling of Empty-answer Queries}. None of the existing approaches offers formal guarantees for correctly identifying empty-answer queries, leaving their reliability in such cases uncertain. How to properly handle them remains overlooked by existing methods and causes inconsistent accuracies: some perform poorly across test cases, while others succeed only on specific datasets. For example, in our experiments, LPLM~\cite{LPLM} identifies at most 25\% of empty-answer queries on DBLP-AN.}



\noindent\textit{G4: Excessive Memory and Computation Overheads.} 
We evaluate both \emph{memory efficiency} (i.e., memory footprint) and \emph{time efficiency}, including the preprocessing time required before supporting queries and the online inference latency. Existing estimators often incur substantial overheads in at least one of these aspects. For example, LPLM~\cite{LPLM} requires approximately 11 hours to prepare its training data, even with 30 parallel processes on the IMDB-MT dataset, and E2E~\cite{E2E} consumes more than 20 MB of memory at inference time.


\vspace{-2mm}
\subsection{\tocheck{Our Proposed Methods and Contributions}}
To address the above research gaps, we introduce \BF{}, a \textbf{L}IKE query \textbf{E}stimator with \textbf{A}ccuracy, \textbf{R}obustness, \textbf{N}egligible overhead, \textbf{T}unability, and \textbf{T}heoretical guarantees. 

\noindent\textbf{\underline{Classification-based Formulation (Sec.~\ref{sec:basic_pattern}).}} For \emph{non-empty-answer queries}, \BF{} reformulates cardinality estimation as a bucket \textbf{\emph{classification problem}}.
Instead of directly regressing the exact cardinality, we partition the cardinality space into a sequence of buckets, each defined by a lower and upper bound and associated with a representative estimate. The estimator then classifies each query into its corresponding bucket. 
This idea allows us to provide \textbf{\emph{formal guarantees on the maximum estimation error}} if we can correctly assign each query into its bucket. We parameterize the bucket boundaries with a user-specified error bound, enabling the control between estimation accuracy and memory usage.

\noindent\textbf{\underline{Bucketed Layered Filter Architecture (Sec.~\ref{sec:bucketed_filter}).}} To assign queries to their corresponding buckets, \BF{} adopts a \textbf{\emph{bucketed layered filter}} architecture (cf. Sec.~\ref{sec:layered-filter}). 
Each bucket contains a multi-layer structure of Bloom filters and a small auxiliary table. 
The Bloom filters serve as compact probabilistic tests for determining whether a query belongs to the bucket, quickly filtering out non-members, while the auxiliary table removes residual false positives. 
This layered design enables accurate bucket assignment with low storage and computation overhead. To further reduce memory consumption, we exploit the empirical skew in query distributions. We avoid constructing explicit filters for the first bucket, which captures the majority of non-empty-answer queries. Additionally, we introduce a \textbf{\emph{frontier-based optimization}} (cf. Sec.~\ref{sec:frontier}), which selectively retains only representative (``frontier'') queries during filter construction, further reducing space without sacrificing accuracy.


\noindent\textbf{\underline{Empty-answer Query Estimation (Sec.~\ref{sec:neg-est}).}} \BF{} also provides dedicated support for \textbf{\emph{empty-answer query estimation}} upon the structure built for non-empty-answer queries. We design two complementary strategies: a direct \emph{filter-based} method and a \emph{prefix-walk refinement} that leverages partial membership checks to detect queries with cardinality equal to $0$ early. Both strategies are formally analyzed, and we derive \textbf{\emph{theoretical bounds on their misclassification probabilities}}, demonstrating that \BF{} achieves conservative yet efficient estimation for empty-answer queries.

\noindent\textbf{\underline{Theoretical Parameter Selection (Sec.~\ref{sec:para_select}).}} We present a theoretical framework for \textbf{\emph{parameter selection}}. We derive expressions for the storage cost in terms of the number of layers and the false positive rates of the Bloom filters, and formulate an optimization problem under user-defined error bound and probability of misclassification of empty-answer queries to minimize memory usage. 

\begin{figure}[t]
    \centering
    \centering
    \includegraphics[width=0.5\textwidth]{ 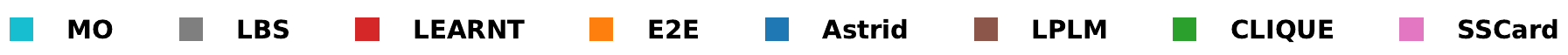}
    \subfloat{
        \includegraphics[width=0.20\textwidth,trim={9 10 10 8}]{ 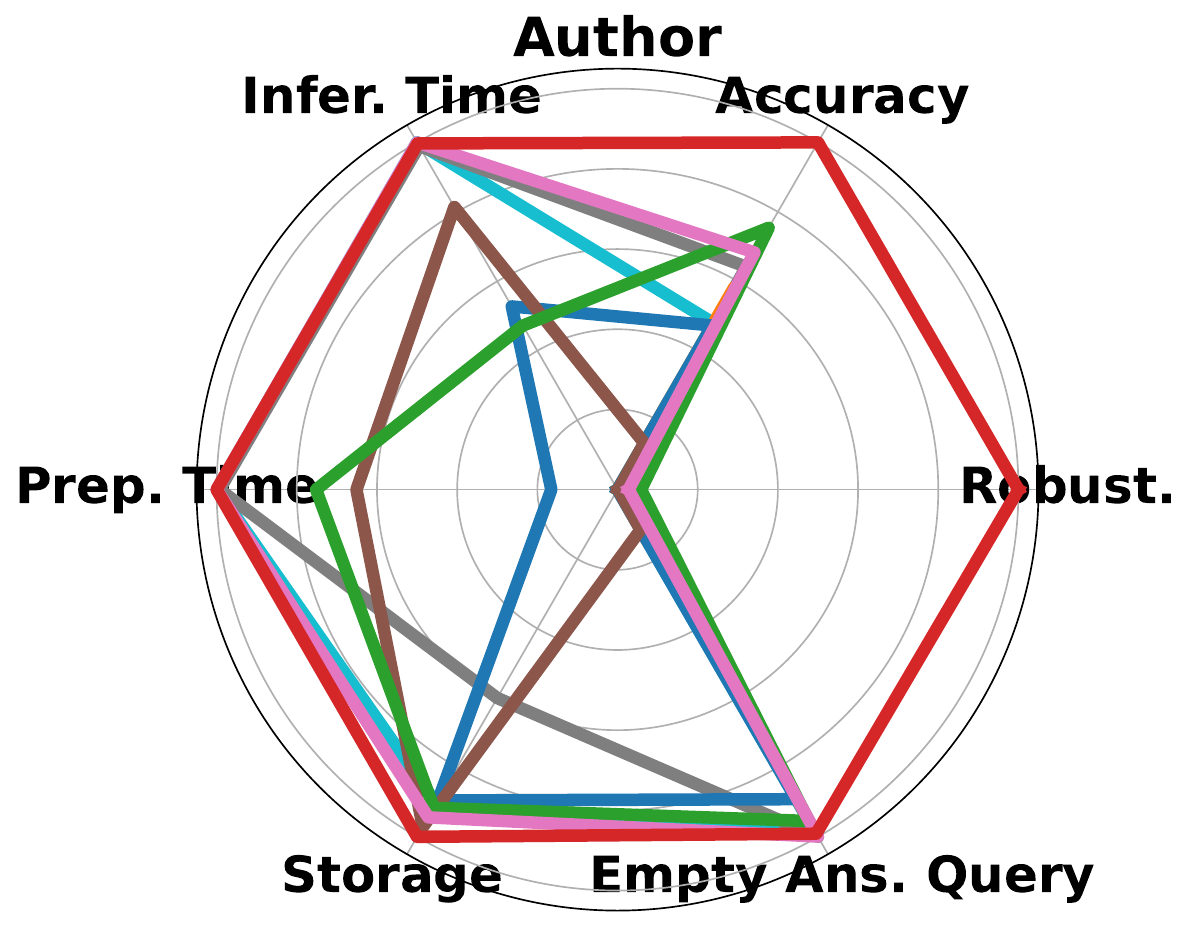}
        }
    ~
    \subfloat{
        \includegraphics[width=0.20\textwidth,trim={9 10 10 8}]{ 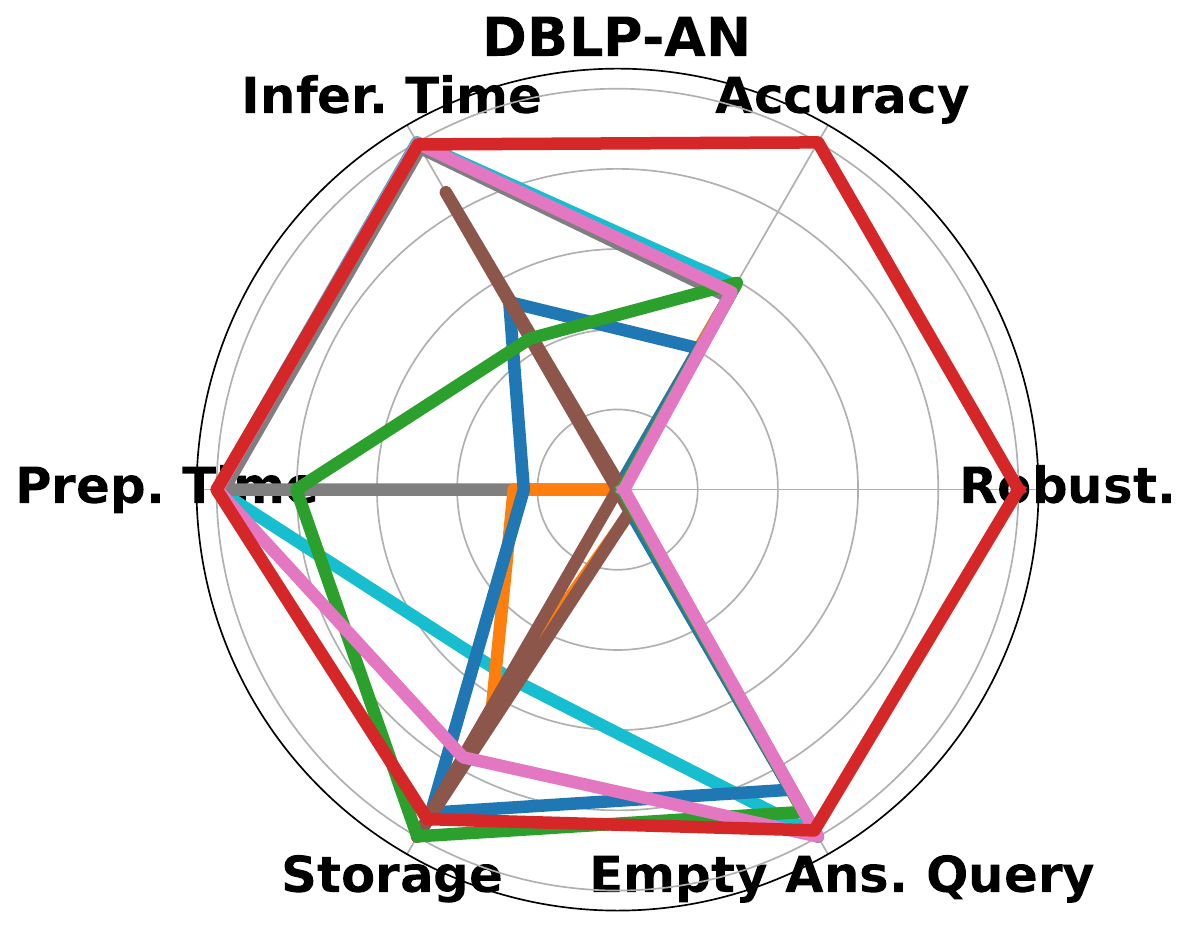}
        }
        \vspace{-1em}
    \caption{
       Comparison of Estimators on Substring Queries 
    }
    \label{fig:spider}
    \vspace{-1em}
\end{figure}

\noindent\textbf{\underline{Long Queries Support (Sec.~\ref{sec:extend_query}).}} The core LEARNT approach is designed for queries with the limited length. Here, we extend \BF{} to support long queries via Markov modeling. By modeling dependencies among consecutive short substrings, this extension reuses \BF{}’s structure to effectively estimate long queries.

\noindent\textbf{\underline{Extensive Experiments (Sec.~\ref{sec:exp}).}} We conduct extensive experiments on four real-world datasets and observe that \BF{} (1) consistently achieves 1.3–1.7× lower mean Q-error than state-of-the-art methods (e.g., CLIQUE~\cite{CLIQUE}), while also significantly reducing tail errors with competitive memory usage; (2) \BF{} requires at most one-third of the construction time of non-learning-based methods and less than 1/70 of the time required by learning-based methods. 

\noindent\textbf{Overall Takeaway.}
To provide an intuitive overview of our method’s strengths, Figure~\ref{fig:spider} compares \BF{} with state-of-the-art estimators across key evaluation dimensions: accuracy, robustness, empty-answer query identification, preprocessing time, inference time, and storage overhead.  Taken together, these results show that \BF{} is not only an accurate
estimator, but also a practical one under deployment-oriented evaluation
criteria. It addresses the key requirements of practical LIKE-query
cardinality estimation in a unified design: accurate estimates for
non-empty-answer queries, formal robustness guarantees, conservative
handling of empty-answer queries, tunable accuracy--overhead trade-offs,
and low construction, storage, and inference overheads. This balance is
important for cost-based optimizers, where cardinality estimators must be
invoked frequently and maintained efficiently as data and workloads evolve.


\vspace{-1mm}
\section{Related Work}\label{sec:related_work}
Existing cardinality estimation on strings (matching \textsf{LIKE} predicates) falls into two categories: learning-based methods and traditional non-learning based methods. 

\noindent\textbf{Learning-based Approaches}. Five recent methods --E2E~\cite{E2E}, Astrid~\cite{Astrid}, CLIQUE~\cite{CLIQUE},  LPLM~\cite{LPLM}, SSCard~\cite{SSCard} -- use learned models. E2E and Astrid are designed to support the three query patterns considered in this study. E2E generates and embeds a set of representative patterns and indexes them with a in-memory trie. For an online query, E2E identifies a stored pattern that covers the query and uses its embedding as input to the estimation model. In E2E, the in-memory trie incurs a substantial memory overhead.

Astrid~\cite{Astrid} provides two variants: \emph{Embed} and \emph{NLM}. The \emph{Embed} variant learns embeddings that queries with similar string patterns and selectivities are mapped to nearby vectors, while \emph{NLM} employs a neural language model to predict the probability of each character conditioned on the preceding characters in the query. Although \emph{Embed} achieves lower Q-error than \emph{NLM}, it incurs significantly higher training cost, as it requires processing a large number of training triplets to learn the embedding space effectively.

In addition to the three fundamental patterns, CLIQUE~\cite{CLIQUE} also supports $\%S\%S\%$, $\%S\%S$, and $S\%S\%$. CLIQUE presents an efficient algorithm to compute exact cardinalities for a given query set. It then trains a learning-based model enhanced with extended $N$-gram tables: for each query, it uses these tables to derive an upper bound\footnote{This upper bound can still be arbitrarily larger than the true cardinality and lacks theoretical guarantees.} on the true cardinality, feeds this bound (along with the query string) into the model, which predicts a coefficient $p\in[0,1]$, and outputs the final estimate as $p$ times the upper bound. 

LPLM~\cite{LPLM} generalizes to queries with any number of `\%’ and `\_’ wildcards by decomposing a query into a sequence of substrings and training a model to predict each substring’s conditional probability. It then multiplies these probabilities by the dataset size to obtain the final cardinality estimate. Despite this flexibility, LPLM underperforms on the most common patterns ($S\%$, $\%S$, $\%S\%$), and incurs high training cost because it must compute true cardinalities for all generated substrings. However, LPLM’s memory footprint is low since it only needs to store the trained model. 

SSCard~\cite{SSCard}, the latest study, is designed specifically for substring queries. It extends the FM-index to support multiple strings and organizes it using a pruned suffix tree, enabling accurate estimation for short patterns and effective compression -- especially on large, skewed alphabets. It also incorporates spline interpolation with error bounds, along with bidirectional estimation and incremental updates to balance accuracy and space efficiency.

\noindent\textbf{Traditional Approaches.} Early estimators build summary structures -- such as suffix trees or pruned suffix trees -- to predict substring selectivity. MO~\cite{MOKVI,MOKVI2, MOKVI3} uses a pruned suffix tree, LBS~\cite{LBS} adds minimal substrings with an $N$-gram table, and P-LSH~\cite{P-LSH} relies on histograms of frequent positional patterns. However, in recent studies, Astrid~\cite{Astrid} and LPLM~\cite{LPLM} show that these methods still could suffer from a large mean Q-error because they cannot capture complex dependencies among characters and substrings. 


\vspace{-2mm}
\section{Recasting Estimation as Classification
}\label{sec:basic_pattern}

The key estimation idea of our approach for \tocheck{\emph{non-empty-answer queries}} lies in transforming the cardinality estimation problem for \textsf{LIKE} queries into a classification task. Specifically, we:
\vspace{-1mm}
\begin{itemize}[leftmargin=*]
    \item \textbf{Partition the cardinality domain} into $n$ buckets \(\{B_i\}_{i=1}^n\), each defined by a lower bound $B_i.c_l$, an upper bound $B_i.c_u$, and a representative estimate $B_i.est$.
    \item \textbf{Classify $Q$} by assigning it to $B_i$ such that the cardinality of $Q$ is in the range, $[B_i.c_l,B_i.c_u]$, and returning \(B_i.est\) as the estimate. 
\end{itemize}

\revision{This classification-based formulation guarantees that, 
\emph{for any query that is correctly assigned to its true bucket} $B_i$, 
the relative estimation error is bounded by
$\max\!\left(\frac{B_i.c_u}{B_i.est}, \frac{B_i.est}{B_i.c_l}\right)$. 
Thus, correctness of bucket classification directly implies 
a formal bound on Q-error.}

\revision{To enable explicit control over the maximum estimation error, 
we introduce a user-tunable error bound parameter $eb > 1$. 
The bucket boundaries are defined as follows:
(1) $B_1.c_l = 1$;
(2) $B_i.c_l = B_{i-1}.c_u + 1$ for $i > 1$;
(3) $B_i.c_u = \lfloor B_i.c_l \cdot eb^2 \rfloor$;
(4) $B_i.est = B_i.c_l \cdot eb$.}

Under the assumption of correct bucket classification, 
the worst-case Q-error of the estimator is exactly $eb$. 
In practice, we set $1 < eb < 2$ to ensure the Q-error remains below 2, 
yielding consistently small estimation errors. 




\vspace{-2.5mm}
\subsection{Potential Solutions and Why They Fall Short}
While viewing the estimation problem as a classification problem is promising, the biggest challenge is \emph{how to reliably assign each query to its true bucket to avoid misclassifications?}

\vspace{-2.5mm}
\subsubsection{Learning-based Approaches} One natural idea is to train a \emph{learning-based classifier}, e.g., an embedding model~\cite{Astrid} or a sequential model~\cite{LPLM,Astrid} to map each query to its bucket. 
However, such models struggle to capture the character dependencies required for accurate bucket assignment.  
Moreover, the query distribution across buckets is highly skewed, making the classification problem imbalanced and training such models particularly challenging.

\begin{table}[!ht]
	\vspace{-3mm}
	\centering
	\caption{Number of Non-empty-answer Queries on  {Author},  {DBLP-AN}, and  {IMDB-MT} with $L=10$.}
	\vspace{-4mm} 
		\begin{tabular}{|c|c|c|c|}
			\hline
			\textbf{Pattern} & \textbf{ {Author}} & \textbf{ {DBLP-AN}} & \textbf{ {IMDB-MT}} \\\hline
                $S\%$ & $815,393$ & $2,145,437$ & $1,108,871$  \\\hline
                $\%S$ & $990,451$ & $1,630,169$ & $1,183,711$\\\hline
                $\%S\%$ & $5,377,686$ & $9,574,392$ & $10,159,610$ \\\hline
		\end{tabular}%
	\label{tab:total-patterns}
	\vspace{-3mm}
\end{table}

\subsubsection{Summary-based Approaches}~\label{sec:exp-bf}  
\tocheck{Ideally, a summary-based approach would index all non-empty-answer queries to provide exact bucket assignments. However, enumerating all possible non-empty-answer queries is impractical due to unmanageable construction and storage overhead. Following existing studies~\cite{LPLM,Astrid,CLIQUE}, we introduce \textbf{a maximum query length $L$}, rendering the query set finite and enumerable offline. While exact summaries, e.g., tries, hash table, are accurate, they incur prohibitive memory costs. Alternatively, we can construct a Bloom filter~\cite{BloomFilter} for each bucket, a more compact approach, but suffers from two drawbacks: 1) \emph{High Memory Footprint.} The filter size scales linearly with the number of non-empty-answer
queries. With roughly 9.56 bits per item at a 1\% false-positive rate, this
design consumes 10.9 MB on the IMDB-AN dataset’s $\%S\%$ pattern based on Table~\ref{tab:total-patterns}. 
(2) \emph{Unresolved False Positives.} False positives in Bloom filter may cause queries to match multiple buckets, resulting in ambiguous classification.}

\subsection{Overall Workflow of \BF{}}
\noindent\revision{To address the limitations of these potential solutions, we propose \BF{}, a bucketed layered filter architecture (Sec.~\ref{sec:bucketed_filter}) based on Bloom filters, to support prefix, suffix, and substring non-empty-answer queries with formal error guarantees for queries of length at most $L$. We further extend \BF{} to handle empty-answer queries (Sec.~\ref{sec:neg-est}) and queries whose length exceeds $L$ (Sec.~\ref{sec:extend_query}).}

\revision{Specifically, given a dataset $\mathcal{S}$, an error bound $e_b$, an optional probability threshold $p_n$ for correctly identifying empty-answer queries, and a query length limit $L$, for each query type:}

\noindent\revision{\textbf{Offline Construction.} \BF{} first \emph{enumerates} all non-empty-answer queries of length at most $L$  in $\mathcal{S}$ and partitions them into buckets according to $eb$, thereby determining the number of buckets and the number of queries in each bucket. It then invokes the \emph{parameter selection} method in Sec.~\ref{sec:para_select} to determine the optimal configuration (e.g., number of layers and Bloom filter false positive rates) and \emph{constructs} the estimator based on Algo.~\ref{alg::fileter_building} in Sec.~\ref{sec:bucketed_filter} accordingly.} 

\noindent\revision{\textbf{Online Estimation.} For a query $Q$, if $|Q| > L$, \BF{} applies the Markov-based extension described in Sec.~\ref{sec:extend_query}. Otherwise, it performs the prefix-walk refinement in Sec.~\ref{sec:enhence-neg}, which handles both empty-answer and non-empty-answer queries in a unified manner. Both process are based on the estimation process in Algo.~\ref{alg::query_classification_basic} in Sec.~\ref{sec:oneline-basic}.}

\vspace{-3mm}
\section{Bucketed Layered Filter: A Memory-efficient Classifier}\label{sec:bucketed_filter}

\tocheck{To ensure the formal error bounds in Sec.~\ref{sec:basic_pattern}, we require precise bucket assignments, which learning-based models struggle to guarantee on skewed distributions. We hence introduce a summary-based classifier that assigns every non-empty-answer query within length $L$ to its correct bucket. We extend support to arbitrarily long queries in Sec.~\ref{sec:extend_query}. Bloom filters offer an excellent starting
point due to their favorable space–accuracy balance, but two challenges emerge, \emph{a large
memory footprint}, and \emph{ unresolved false
positives}. We address the first challenge by exploiting the skew
in query distributions (Sec.~\ref{sec:skip-b1}), storing far fewer entries
than a naive design. To eliminate false positives, we adopt a layered
filter with a small lookup table (Sec.~\ref{sec:layered-filter}). We further cut storage by leveraging relationships among
queries (Sec.~\ref{sec:frontier}), while preserving correct classification.}

\begin{table}[!ht]
	\small
	\vspace{-2mm}
	\centering
	\caption{Percentages (\%) of Non-empty-answer Queries in Each Cardinality Range.} 
	\vspace{-4mm} 
	\resizebox{0.48\textwidth}{!}{%
		\begin{tabular}{|c|c|c|c|c|c|c|}
			\hline
			\textbf{Type} & \textbf{[1,2]} & \textbf{[3,6]} & \textbf{[7,15]} & \textbf{[16,36]} & \textbf{[37,83]} & \textbf{[84,)} \\\hline
                 {Author}-$S\%$&80.13&13.3&4.31&1.33&0.53&0.41\\\hline
                 {Author}-$\%S$&79.63&14.34&4.31&1.14&0.36&0.22\\\hline
                 {Author}-$\%S\%$&76.19&15.83&5.35&1.62&0.57&0.44\\\hline
                 {DBLP-AN}-$S\%$&72.62&17.15&6.46&2.37&0.81&0.59\\\hline
                 {DBLP-AN}-$\%S$&72.43&17.65&6.57&2.25&0.69&0.4\\\hline
                 {DBLP-AN}-$\%S\%$&69.34&18.61&7.51&2.87&0.99&0.69\\\hline
                 {IMDB-MT}-$S\%$&86.54&8.73&2.87&1.16&0.43&0.28\\\hline
                 {IMDB-MT}-$\%S$&87.39&8.19&2.72&1.04&0.41&0.25\\\hline
                 {IMDB-MT}-$\%S\%$&85.38&9.19&3.19&1.31&0.54&0.4\\\hline
                		\end{tabular}%
	}
	\label{tab:card-distribution}
\end{table}

\vspace{-4mm}
\subsection{Skipping Filter Construction Motivated by Query Skew}\label{sec:skip-b1}


A major issue with Bloom filters is their large memory cost when many non-empty-answer queries exist. Can we reduce the number of queries used in filter construction? Our empirical study on common string datasets shows that 
this is indeed feasible.
{\revision{{Table~\ref{tab:card-distribution}}}}\footnote{\revision{To produce the empirical query cardinality distributions, we enumerate all prefix, suffix, and substring with length at most $L$ for each dataset and compute their exact cardinalities by counting matching records. The resulting counts are aggregated to form the empirical distributions reported.}} shows the percentages of queries in each predefined cardinality ranges ($eb$ with $1.5$) under three datasets and different query patterns. It reveals the following key observation:

\vspace{-1mm}
\begin{obsv}\label{sec:o1}
   The query distribution (the number of queries in each cardinality range) across datasets and query patterns is highly skewed, with the majority of queries exhibiting low cardinality.
\end{obsv}
\vspace{-1mm}

This skewed distribution simplifies filter construction. Since most
queries lie in $[1, 2]$, i.e., $B_1$, building a layered filter for $B_1$ is
unnecessary. We can instead build filters only for buckets $B_i$ ($i > 1$). During
classification, if a query matches none of these filters, we assign it
to $B_1$ with estimate $eb$. This greatly reduces memory while preserving
classification accuracy.

\begin{figure}[t]
	\centering
	\includegraphics[width=0.5\textwidth]{ 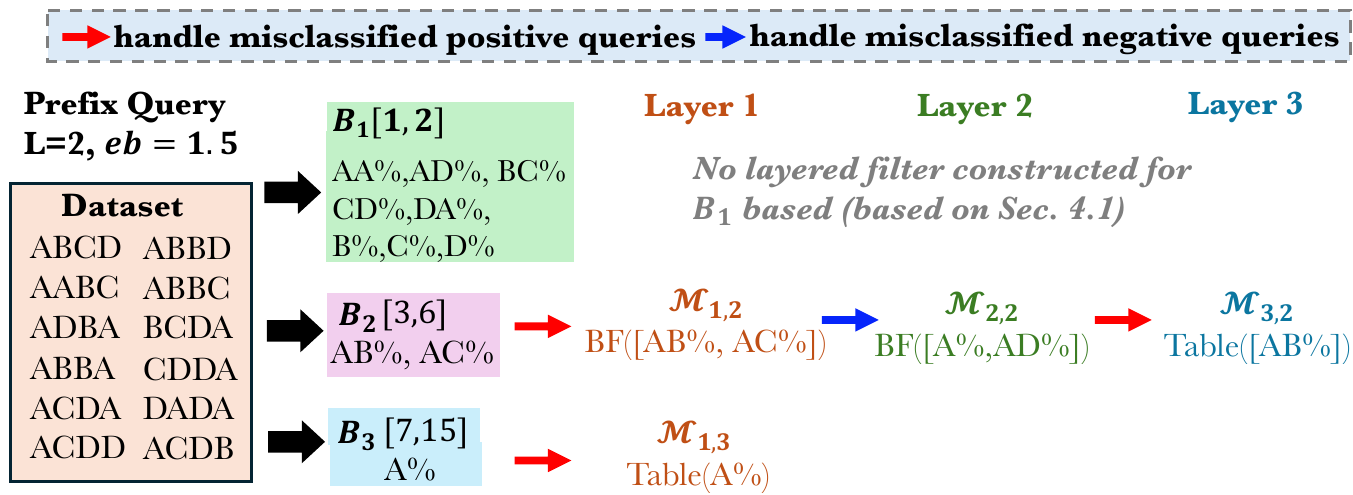}
	\vspace{-6mm}
	\caption{\revision{Example of Bucketed Layered Filter}}
	\label{fg:building-ov}
	\vspace{-4mm}
\end{figure}

\subsection{A Layered Filter for Each Bucket}\label{sec:layered-filter}
\subsubsection{\tocheck{Idea of A Layered Filter}} 

\revision{Figure~\ref{fg:building-ov} illustrates the bucketed layered filter structure and construction using a concrete prefix-query example with $L=2$ and $eb=1.5$. The enumerated queries are partitioned into three buckets, $B_1$, $B_2$, and $B_3$. 
The first bucket $B_1$ is skipped based on Sec.~\ref{sec:skip-b1} and we focus on building the filter for $B_2$ here. 
The goal is to ensure that all queries in $B_2$ return \textsf{true}, while queries from other buckets return \textsf{false}.}

\noindent\revision{\textbf{Layer 1.} We first build a Bloom filter on the positive set $B_2.K = \{\text{AB\%, AC\%}\}$ and store it in $\mathcal{M}_{1,2}$ ($\mathcal{M}_{j,i}$ stores the data structure of $j$-layer of $B_i$). This guarantees that all $B_2$ queries pass Layer~1.}

\noindent\revision{\textbf{Layer 2.}  However, false positives in Bloom filter may allow some queries from $B_1$ or $B_3$ to pass $\mathcal{M}_{1,2}$ as well. We evaluate all negative queries (i.e., queries from $B_1$ and $B_3$)
against $\mathcal{M}_{1,2}$ and collect those misclassified as positive. They are $A\%$ and $AD\%$ in the example.  We then build a second-layer structure to capture these misclassified negative queries. Here we build another bloom filter stored in $\mathcal{M}_{2,2}$.}

\noindent\revision{\textbf{Layer 3.} Some queries from $B_2$ may be mistakenly rejected (false negatives) as they may be incorrectly classified as positive in $\mathcal{M}_{2,2}$, which is built on the queries not in $B_2$. In the example, $AB\%$ is classified in $\mathcal{M}_{2,2}$.  
To correct this, we add a third layer that records $AB\%$ by building up a lookup table stored in $\mathcal{M}_{3,2}$.}

\noindent\revision{\textbf{Query Example.} Consider the query AB\%, which belongs to bucket $B_2$. It first passes Layer~1 ($\mathcal{M}_{1,2}$), which is built on $B_2$’s positive keys. However, it incurs a false positive at Layer~2 ($\mathcal{M}_{2,2}$), which is constructed to capture misclassified negative queries. This would incorrectly reject $AB\%$. The final lookup table ($\mathcal{M}_{3,2}$) corrects this by explicitly storing such remaining positive queries, ensuring AB\% is correctly classified into $B_2$.}

\noindent\emph{\underline{Discussion.}} The above example consists of three layers in $B_2$, but the process can be extended with additional alternating Bloom filters. \revision{In general, odd-numbered layers are built on the currently misclassified positive queries, while even-numbered layers are built on the currently misclassified negative queries. The final layer is always implemented as a deterministic lookup table. This design is crucial: although intermediate Bloom filters may introduce false positives, the final lookup table explicitly corrects any remaining misclassifications, ensuring exact bucket assignment. Consequently, the error bound derived in Sec.~\ref{sec:basic_pattern} is guaranteed to hold.}

In Sec.~\ref{sec:filter-build}, we formalize the layered architecture for a general bucket $B_i$ and present its systematic construction algorithm. The number of Bloom-filter layers is determined in Sec.~\ref{sec:para_select}, where we develop a principled parameter selection strategy. The formal online query process is described in Sec.~\ref{sec:oneline-basic}.

\setlength{\textfloatsep}{0pt}
\begin{algorithm}[t]
	\small
	\caption{Bucketed Layered Filter Building Process}
	\label{alg::fileter_building}
	\KwIn{bucket set ${B_i}$ ($i\in[1,n]$), number of layers $m$, false positive rates $f_j$ ($j\in[1,m-1]$)}  
	\KwOut{Bucketed Layered Filter $\mathcal{M}$}  

        $\mathcal{M} \gets []$\;
        \ForEach{$i \in [2, 3, ..., n]$}{
            $\mathcal{Q}_{neg} \gets []$\;
            \ForEach{$j \in [1, 2, ..., n]$}{
                \If{$j > i$ or $j = 1$} {
                    $\mathcal{Q}_{neg}.add(B_j.K)$\;
                }
            }
            $\mathcal{Q}_{pos} \gets B_i.K$\;
            \ForEach{$j \in [1, 2, ..., m-1]$}
            {
                $\mathcal{Q}_{mis} \gets []$\;
                \If{$j \text{ mod } 2=1$}{
                    $\mathcal{M}_{j,i} \gets $ BuildBloomFilter($\mathcal{Q}_{pos},f_j$)\;
                    \ForEach{$Q \in \mathcal{Q}_{neg}$} {
                        \If{$Q \in \mathcal{M}_{j,i}$}{
                            $\mathcal{Q}_{mis}.add(Q)$\;
                        }
                    }
                    $\mathcal{Q}_{neg} \gets \mathcal{Q}_{mis}$\;
                }\Else{
                    $\mathcal{M}_{j,i} \gets $ BuildBloomFilter($\mathcal{Q}_{neg},f_j$)\;
                    \ForEach{$Q \in \mathcal{Q}_{pos}$} {
                        \If{$Q \in \mathcal{M}_{j,i}$}{
                            $\mathcal{Q}_{mis}.add(Q)$\;
                        } 
                    }
                    $\mathcal{Q}_{pos} \gets \mathcal{Q}_{mis}$\;
                }
            
            }
            \If{$m \text{ mod } 2 = 1$}{
                $\mathcal{M}_{m,i} \gets $ BuildTable($\mathcal{Q}_{pos}$)\;
            }\Else{
                $\mathcal{M}_{m,i} \gets $ BuildTable($\mathcal{Q}_{neg}$)\;
            }
        }
        return $\mathcal{M}$
\end{algorithm}

\subsubsection{Bucketed Layered Filter Building}~\label{sec:filter-build}
Algo.~\ref{alg::fileter_building} presents the  process of building the bucketed layered filter. Suppose the two-dimensional array, $\mathcal{M}$ with $m$ rows and $n$ columns, is used to store the bucketed layered filter. For each bucket $B_i$, we build a layered filter with $m$ ($m \geq 2$) layers. $\mathcal{M}_{j,i}$ refers to the summary (either a Bloom filter or a lookup table) of bucket $B_i$ at the $j$-th layer. Let $f_j$ denote the false positive ratio of $j$-th layer.\footnote{For simplicity, we assume that the same layer of different buckets use the same false positive ratio. This assumption is lifted in our parameter selection process.} For each $B_i$ ($i >  1$), we build the layered filter with following steps:

\smallskip\noindent\underline{\textit{Step 1: Collect Negative and Positive Queries}}. To avoid classifying the queries not in $B_i$ with $B_i$, we collect queries from other buckets and check their classification during our building process. We call them as \emph{negative} queries for $B_i$. 
Based on Observation~\ref{sec:o1}, we design an online classification process that starts from $B_2$, i.e., from the buckets with a large number of non-empty-answer queries. Hence, the queries in $B_j$ ($j < i \text{ and } j \ne 1$) will never be misclassified in the layered filter of $B_i$. We only collect queries from $B_1$ and $B_j$ $(j>i)$ and store them in $\mathcal{Q}_{neq}$ (Lines 2-6) and we initialize $\mathcal{Q}_{pos}$ with $B_i.K$ (Line 7). In what follows, $\mathcal{Q}_{neq}$ and $\mathcal{Q}_{pos}$ store the currently misclassified negative queries and the positive queries, respectively.

\smallskip\noindent\underline{\textit{Step 2: Build Bloom Filters}}. Next, we proceed to build $m-1$ Bloom filters. The filter at each odd-numbered layer is built using the currently misclassified positive queries, i.e., $\mathcal{Q}_{pos}$ (Line 11) while the filter at each even-numbered layer is built using the currently misclassified negative queries, i.e., $\mathcal{Q}_{neq}$ (Line 17). 

After building the filter at each layer, we evaluate the effectiveness of the filter by checking which queries from the opposing class, $\mathcal{Q}_{neq}$ or $\mathcal{Q}_{pos}$, incorrectly pass through the filter, i.e., the currently misclassified queries (Lines 12-14 and Lines 18-20). They are stored in $\mathcal{Q}_{mis}$, which becomes the new $\mathcal{Q}_{neq}$ or $\mathcal{Q}_{pos}$ for the next layer.

\smallskip\noindent\underline{\textit{Step 3: Construct the Lookup Table}}. The final layer is a built lookup table rather than a Bloom filter to ensure precise membership decisions. If $m$ is odd, the table is constructed from the remaining $\mathcal{Q}_{pos}$; otherwise, it is built from $\mathcal{Q}_{neg}$.


\begin{algorithm}[t]
	\small
	\caption{Query Classification}
	\label{alg::query_classification_basic}
	\KwIn{query $Q$, a bucketed layered filer $\mathcal{M}$, bucket number $n$, layer number $m$} 
	\KwOut{bucket ID that $Q$ belongs to}  
        \ForEach{$i \in [2, 3, ..., n]$}{
                \ForEach{$j \in [1, 2, ..., m-1]$}{
                    \If{$Q \notin \mathcal{M}_{j,i}$} {
                        \If{$j$ mod $2=1$}{
                        break; \Comment{$Q$ is an invalid query for $B_i$.}
                        }\Else{
                            \Return $i$\;
                        }
                    }
                }

                \If{$Q \in \mathcal{M}_{m,i} \text{ XOR } (m \text{ mod }2=1)$} {
                    \Return $i$\;
                } 
                
        }
        \Return $1$\;
        
\end{algorithm}

\subsubsection{Online Classification Process}~\label{sec:oneline-basic}
The pseudo-code is presented in Algo.~\ref{alg::query_classification_basic}. Given a query $Q$, we sequentially check whether $Q$ belongs to a bucket $B_i$ in ascending order, starting from $B_2$, followed by $B_3, B_4$, and so on. For each $B_i$, we check $Q$ layer by layer. Since the odd-numbered layers and even-numbered layers are built upon positive queries and negative queries, respectively, if $Q$ is not in an odd-numbered layer, it must be a negative query and we go to check next bucket. Similarly, it $Q$ is not in an even-numbered layer, it must be a positive query and we return $i$. If $Q$ is filtered out from all buckets $B_i$ for $i \ge 2$, the algorithm returns 1.

\subsection{Frontier-based Pruning}\label{sec:frontier} 
Recall that to guarantee correct bucket ID predictions for all queries, the construction of the \emph{second-layer} Bloom filter in each bucket $B_i$ ($i \geq 2$) requires checking for potential false matches against the queries in bucket $B_1$ (Lines 5-6 in Algo.~\ref{alg::fileter_building}). Because $B_1$ can be large, false positives at the first-layer Bloom filter may cause a substantial number of its queries to be misclassified. This enlarges the candidate set for the second layer, requiring it to be built over a much larger query set and thus incurring higher memory cost.

\smallskip\noindent\textbf{Construction Process}. We propose a \textbf{frontier-based} pruning strategy that selectively uses a subset of $B_1$’s queries when constructing the Bloom filter for each $B_i$. Specifically, we define frontier queries ($B_1.FQ$) as those that are not covered by any other query in $B_1$. Formally, a query $Q \in B_1.K$ is a frontier query if there does not exist another $Q^\prime \in B_1.K$ such that $Q’$ is a prefix of $Q$. \revision{For example, $B_1.K$ in Figure~\ref{fg:building-ov} is $\{B\%, C\%, D\%, AA\%, AD\%, BC\%, CD\%, DA\%\}$, then $B_1.FQ = \{B\%, C\%, D\%, AA\%, AD\%\}$. In Algo.~\ref{alg::fileter_building},  this pruning is applied by replacing $B_1.K$ with $B_1.FQ$ in Line 6 when $j = 1$, i.e., reducing the number of queries in $B_1$ being checked.}

\smallskip\noindent\textbf{Online Classification}. Using only frontier queries introduces a challenge during prediction. If a query $Q \in B_1$ is not a frontier query, it could be misclassified, since such a query is not explicitly considered during the filter construction process. To mitigate this, we refine the online prediction process. Instead of evaluating only $Q$ itself (as in Sec.~\ref{sec:oneline-basic}), we perform a \textbf{prefix walk}, which iteratively checks the bucket IDs associated with $Q$ and its prefixes. If any prefix of $Q$ is classified into $B_1$, we assign $Q$ to $B_1$. Otherwise, we select the bucket with the smallest predicted ID.

\revision{In Figure~\ref{fg:building-ov}, consider a query $Q = BC\%$. By checking each prefix of it, $\{BC\%, B\%\}$, based on Algo.~\ref{alg::query_classification_basic}, we find that $B\%$ maps to $B_1$. Thus, $Q$ is predicted to belong to $B_1$.}

\section{Classifying Empty-answer Queries with Theoretical Guarantees}\label{sec:neg-est}
We now describe how our estimator handles \emph{empty-answer} queries. 

\noindent\revision{\textbf{{Rationale for Mapping Empty-answer Queries to $B_1$.}} Accurately identifying that a query has a cardinality of zero is
intrinsically difficult. The space of potential empty-answer queries
is substantially larger than that of non-empty-answer queries, and
complete detection would require enumerating a prohibitively large set of patterns, which is
infeasible. Moreover, practical query optimizers, e.g. PostgreSQL, rarely predict a cardinality
of zero unless emptiness can be proven (e.g., contradictory
predicates). Assigning zero without certainty may prematurely
eliminate execution paths or suppress necessary joins, leading to
severely suboptimal plans.}

\revision{In our framework, $B_1$ represents the smallest non-zero
cardinality range (e.g., $[1,2]$ when $eb = 1.5$). Mapping an
empty-answer query to $B_1$ therefore introduces only a minimal,
bounded overestimate. Compared to assigning a larger bucket,
this conservative choice limits the impact on plan generation
while preserving the lightweight design of the estimator.}

Motivated by this, we assign empty-answer queries to bucket~$B_1$ and develop two strategies: a naive filter-based method (Sec.\ref{sec:neg_naive}) and a prefix-walk refinement (Sec.\ref{sec:enhence-neg}), both of which integrate naturally with our layered filter framework and provide formal guarantees.


\subsection{Filter-Based Classification}\label{sec:neg_naive}

Let $Q_{\text{zero}}$ denote an empty-answer query. Assume we have $n$ buckets, $B_1, \dots, B_n$, where each $B_i$ ($i \ge 2$) is equipped with a layered filter of $m$ layers. Each layer $j$ has a false positive rate $f_j$, and the final layer is a lookup table. During classification, we sequentially probe the buckets from $B_2$ to $B_n$, as Sec.~\ref{sec:oneline-basic}. \revision{{Under the assumption of independence between false-positive events across layers, 
which follows from using independent hash functions for each Bloom filter layer.}}\footnote{\revision{Each Bloom filter layer in each bucket is instantiated using an independent universal hash function drawn from a hash family with pairwise independence. As a result, membership tests across layers and buckets correspond to independent Bernoulli trials.}} \revision{We derive the probability that an empty-answer query is assigned to each bucket, formalized in Lemma~\ref{lem:neg-misclass}.}

\begin{lemma}[Classification Probability of Empty-answer Queries into \texorpdfstring{$B_i$}{Bi}]\label{lem:neg-misclass}
Under the assumption of independent false-positive events, the probability $p_i$ of $Q_{\text{zero}}$ assigned to bucket \(B_i\) is
\[
  p_i
  = \sum_{k=1}^{\lfloor\frac{m-1}{2}\rfloor}\prod_{j=1}^{2k-1} f_j\cdot(1-f_{2k})  + \mathds{1}_{\{m \text{ mod }2=0\}}\cdot\prod_{j=1}^{m-1}f_j
  \quad i\ge2,
\]
and the fall‐through to \(B_1\) is 
$
  \Pr[Q_{\text{zero}}\to B_1]
  = \prod_{i=2}^{n}(1-p_i)
$.
\end{lemma}

\begin{proof}
There are two main cases that $Q_{\text{zero}}$ will be classified into $B_i$: (1) it is not in any odd-numbered layer -- the probability is $\sum_{k=2}^{\lfloor\frac{m-1}{2}\rfloor}\prod_{j=1}^{2k-1} f_j\cdot(1-f_{2k})$; (2) if the last Bloom filter is at an even-numbered layer, i.e. $m \text{ mod }2=0$, and $Q_{\text{zero}}$ is in that filter -- the probability is $\prod_{j=1}^{m-1}f_j$. Thus, the probability that $Q_{\text{zero}}$ will be classified in $B_i$ is the sum of them. 

For each \(i\ge2\), let \(E_i\) be the event “$Q_{zero}$ is classified in $B_i$”. So \(\Pr[E_i]=p_i\). By sequential probing and independence,
\[
  \Pr[Q_{\text{zero}}\to B_i]
  = \Pr\bigl(\neg E_2\wedge\cdots\wedge\neg E_{i-1}\bigr)\,\Pr[E_i]
  = \Bigl(\prod_{j=2}^{i-1}(1-p_j)\Bigr)\,p_i,
\]
and if we end up in \(B_1\), giving the product over \(j=2..n\).  
\qedhere
\end{proof}


To simplify analysis, we assume all layers in $B_i$ share the same false positive rate $f$. The misclassification probability becomes:
\begin{equation}\label{eq:neg_p_pw}
p_i = 
\begin{cases}
\frac{f-f^m}{1+f} & \text{if } m \text{ mod }2= 1 \\
\frac{f+f^m}{1+f} & \text{if } m \text{ mod }2= 0
\end{cases}
\end{equation}

\subsection{Prefix-Walk Refinement}\label{sec:enhence-neg}
To further improve robustness, we introduce a \emph{prefix-walk refinement} strategy that examines multiple prefixes of the query. This approach increases the likelihood of correctly assigning $Q_{\text{zero}}$ to $B_1$ by detecting inconsistencies across prefix classifications.

\revision{Let $l$ be the length of $Q_{\text{zero}}$. We define a sequence of $t = l - l_1 + 1$ prefixes:
$
Q_0 = Q_{\text{zero}}, Q_1 = Q_{\text{zero}}[:-1],  \dots, Q_{t-1} = Q_{\text{zero}}[:l_1], $
where $Q_{\text{zero}}[:l_1]$ is the longest known non-empty-answer prefix. Each prefix $Q_k$ is classified into a bucket ID $b_k \in \{1, \dots, n\}$ using Algo.~\ref{alg::query_classification_basic}. Since $Q_i$ is a prefix of $Q_{i-1}$,
its true cardinality must be no smaller than that of $Q_{i-1}$, i.e., $b_i \ge b_{i-1}$, and $b_0$ must be the smallest among all $b_i$. Based on this \emph{monotonicity property}, we determine the bucket assignment:
\begin{itemize}[leftmargin=*]
  \item If any $b_k = 1$, assign $Q_{\text{zero}}$ to $B_1$. Since our estimator never misclassifies non-empty-answer queries, the $b_i = 1$ guarantees that $Q$ must belong to $B_1$ based on the \emph{monotonicity property}.
  \item If there exists an index $i$ such that $b_i < b_{i-1}$, we assign the query to $B_1$. $b_i < b_{i-1}$ violates \emph{monotonicity property}, i.e., at least $Q_{i-1}$ is assigned to wrong buckets. Since non-empty-answer queries are never misclassified, this implies that $Q_{i-1}$ is empty-answer; hence $Q_{{zero}}$ must be assigned to $B_1$.
  \item Otherwise, assign to $b_0$, the smallest bucket id.
\end{itemize}
}
This rule enhances robustness by exploiting prefix-based signals indicative of invalidity. \revision{{In the following, we present a probabilistic analysis of correctly classifying an empty-answer query under the assumption that each prefix $Q_k$ is classified independently.}}\footnote{\revision{Although prefix queries are structurally related, this assumption holds because each prefix is mapped to hash values using independent hash functions. Even a single-character change between successive prefixes (e.g., from $Q_i$ to $Q_{i-1}$) results in an uncorrelated hash output under universal hashing. Consequently, membership tests for different prefixes are independent trials with respect to Bloom filter false positives.}}

\begin{lemma}[Prefix-Walk Classification Probability of Empty-answer Queries into \texorpdfstring{$B_1$}{B1}]
\label{lem:prefix-negative}
Assume each prefix $Q_k$ is classified independently, and the probability of $b_k=1$ is $q$.
Let $t$ be the number of prefixes. Then the probability of $Q_{zero}$ classified into $B_1$ is
\[
\Pr[Q_{\text{zero}}\to B_1] = 1 - (1 - q)^t \cdot \frac{\binom{n + t - 2}{t}}{(n - 1)^t}.
\]
\end{lemma}

\begin{proof}
Let $b_0, b_1, \dots, b_{t-1}$ be the bucket IDs assigned to prefixes $Q_0, \dots, Q_{t-1}$. The method outputs 1 if either of the following occurs:
\begin{enumerate}
  \item At least one prefix is classified into bucket 1: $\exists k: b_k = 1$.
  \item The sequence contains a decrease: $\exists i \in [1, t-1]: b_i < b_{i-1}$.
\end{enumerate}

We compute the complement event: 
\begin{itemize}
  \item All $b_k > 1$, which occurs with probability $(1 - q)^t$.
  \item The sequence $b_0, \dots, b_{t-1}$ is non-decreasing over $\{2, \dots, n\}$.
\end{itemize}

Let $S$ be the number of non-decreasing sequences of length $t$ over $\{2, \dots, n\}$, i.e., over a support of size $(n - 1)$. This is a standard stars-and-bars problem~\cite{stars-bars}:
$
S = \binom{(n - 1) + t - 1}{t} = \binom{n + t - 2}{t}.
$

The total number of possible sequences with values in $\{2, \dots, n\}$ is $(n - 1)^t$, assuming independence. Therefore, the probability of a non-decreasing sequence conditioned on all $b_k > 1$ is: $\frac{\binom{n + t - 2}{t}}{(n - 1)^t}.$

Multiplying this by the probability that all $b_k > 1$, we obtain the total probability of not predicting bucket 1:
$
\Pr[Q_{\text{zero}}\not\to B_1] = (1 - q)^t \cdot \frac{\binom{n + t - 2}{t}}{(n - 1)^t}.
$ Hence, the desired probability is: $
\Pr[Q_{\text{zero}}\to B_1] = 1 - (1 - q)^t \cdot \frac{\binom{n + t - 2}{t}}{(n - 1)^t}. $
\end{proof}

$q$ above is also the probability of an empty-answer query classified with $B_1$ in Lemma~\ref{lem:neg-misclass}. Here, we prove we can achieve a higher probability in Lemma~\ref{lem:prefix-negative}. When $t\ge2$, $\frac{\binom{n + t - 2}{t}}{(n - 1)^t} < 1$. Thus, $1 - (1 - q)^t \cdot \frac{\binom{n + t - 2}{t}}{(n - 1)^t} \ge 1 - (1 - q)^t \ge 1-(1-q)^2$. When $q \in (0,1)$, $1-(1-q)^2 > q$. Thus, we have $1 - (1 - q)^t \cdot \frac{\binom{n + t - 2}{t}}{(n - 1)^t} > q$, i.e., achieving a larger probability to classify $Q_{zero}$ with $B_1$. 

Note $l_1$ is used solely for theoretical analysis and is not needed during online estimation. This is because for any prefixes $Q_i$ and $Q_j$, if $i < j$, then the corresponding bucket IDs satisfy, $b_i \le b_j$.

\vspace{-1mm}
\section{Parameter Selection}\label{sec:para_select}
One of the key advantages of \BF{} is its tunability: users can control both estimation accuracy and system overhead via two intuitive parameters, the desired error bound $eb$ for non-empty-answer queries and the minimum probability $p_n$ of correctly identifying empty-answer queries. Achieving this tunability requires configuring  \BF{}, which is defined by parameters such as the number of buckets, the number of filter layers, and the false positive rates of the Bloom filters. In this section, we present a formal framework to guide parameters selection. 

\vspace{-2mm}
\subsection{Cost Modeling}\label{sec:storage_ana}
We begin by analyzing the storage cost for a single bucket $B_i$. Let $N_p = |B_i.K|$ denote the number of positive (in-bucket) queries and $N_n = \sum_{j\in[1, 2, ..., n]\cap(j>1 \cup j =1)}|B_j.K|$ denote the number of negative queries (from other buckets) considered during filter construction for $B_i$. Let $L$ denote the maximum length of a query. Thus, a query in the lookup table takes $8L$ bits.

The storage usage of $B_i$ consists of three parts: (1)  Bloom filters in odd-number layers:
$\sum_{i=1}^{\lfloor\frac{m}{2}\rfloor}\frac{-\text{ln}f_{2i-1}}{(\text{ln}2)^2} N_p\prod_{j=1}^{i}f_{2j-2}$;
(2) Bloom filters in even-number layers: $\sum_{i=1}^{\lfloor\frac{m-1}{2}\rfloor}\frac{-\text{ln}f_{2i-1}}{(\text{ln}2)^2} N_n \prod_{j=1}^{i}f_{2j-1}$
(3) Lookup table in the last layer: $N_p\prod_{j=1}^{\lfloor\frac{m}{2}\rfloor}f_{2j}8L$ if $ m \text{ mod } 2 = 1$, otherwise $N_n\sum_{j=1}^{\lfloor\frac{m}{2}\rfloor}f_{2j-1}8L$. 

For simplicity in our analysis, we assume all Bloom filter layers within a bucket share the same false positive rate, $f$. Given this, we can express the storage cost with:

\noindent If $m \text{ mod } 2 =1$, let $m=2k+1$ ($k \ge 1$) and we have:
\begin{align}\label{eq:s-odd}
  S_{\text{odd}}(k,f)= \frac{-\text{ln}f}{(\text{ln}2)^2}\frac{N_p+N_nf}{1-f}(1-f^k) + 8LN_pf^k  
\end{align}

\noindent And if $m \text{ mod }2=0$, let $m=2k$ ($k>1$) and we have:
\begin{align}\label{eq:s-even}
  S_{\text{even}}(k, f) = \frac{-\text{ln}f}{(\text{ln}2)^2}\frac{1}{1-f} \left[ N_p (1 - f^{k}) + N_n f (1 - f^{k-1}) \right] + 8 L N_n f^{k}  
\end{align}

The total storage is the sum of them over all buckets, $B_i$ ($i \ge 2$).

\subsection{Parameters Optimization}
\subsubsection{Optimizing Parameters Without Empty-answer Query Constraints}\label{sec:para_select_basic}
Our primary goal is to select the number of layers $m$ (or equivalently, $k$ in Equation~\ref{eq:s-odd} and ~\ref{eq:s-even}) and the false positive rate $f$ in Bloom filters to minimize the total storage cost. 
Once the user specifies the desired error bound $eb$, the number of buckets $n$ is determined. The values of $N_p$, $N_n$, and $L$ for each $B_i$ are fixed. 

However, this presents a mixed-integer optimization problem as in Equation~\ref{eq:s-odd} and ~\ref{eq:s-even}, $f$ is continuous while $k$ is restricted to integers. To solve this efficiently, we adopt a practical, iterative approach:

\noindent \textbf{(1) Enumerate Layer}: We iterate through a small set of integer values for the number of layers, $m$, starting from 2 and set $k=\lfloor\frac{m}{2}\rfloor$. And based on our empirical analysis, only a few values of $m$ suffice.

\noindent \textbf{(2) Find Optimal $f$}: For each chosen $m$ ($k$), we find the optimal false positive rate, $f$, that minimizes the storage cost function. $S_{odd/even}$ is a contentious function when $f\in(0,1)$. Thus, the minimized value of $S_{odd/even}$ is obtained at the root of $S'(k,f)=0$ or at the end point. However, setting $S’(f)=0$ mixes $\ln f$ and $f^{k}$ terms -- no closed-form root for generic $k$, $N_1$, $N_2$ and it even has more than one root. Thus, we employ DIRECT~\cite{DIRECT}, a deterministic, derivative-free method for global optimization over bounded domains. 

\noindent \textbf{(3) Select Best Configuration}:  We choose the $(m, f)$ pair that results in the minimum storage cost across all tested values of $m$.

\vspace{-1mm}
\subsubsection{Optimizing Parameters with Empty-answer Query Guarantees}
To enhance accuracy in identifying empty-answer queries, users can specify an additional constraint: a lower bound $p_n$ on the probability that an empty-answer queries is classified into bucket $B_1$. This constraint restricts the feasible range for the false positive rate $f$.

Based on Lemma~\ref{lem:prefix-negative}\footnote{The analysis for Lemma~\ref{lem:neg-misclass} follows a similar process but is more straightforward, so we omit it.}, we derive the constraint:
 $1 - (1 - q)^t \cdot \frac{\binom{n + t - 2}{t}}{(n - 1)^t}\ge p_n$ and we get the condition on $q$ that $q \ge 1-(n-1)\left(\frac{1-p_n}{\binom{n+t-2}{t}}\right)^{\frac{1}{t}}$. For a given $eb$, $n$ is determined and is a constant here while different empty-answer queries may have different values of $t$. Let $g(t)=\left(\frac{1-p_n}{\binom{n+t-2}{t}}\right)^{\frac{1}{t}}$. $g(t)$ increases as $t$ increases. To support different $t$s, we therefore set $q \ge 1-(n-1)g(2)$.

Let $c=1-(n-1)g(2)$ and we have the constraint: $q\ge c$. 
Based on Lemma~\ref{lem:neg-misclass} and Lemma~\ref{lem:prefix-negative}, we have $q=\prod_{i=2}^n(1-p_i)$, where $p_i$ is the probability that a query is a false positive in $B_i$ ($i \ge 2$).  
For simplicity in our analysis, we assume that $p_i$ is the same cross different buckets. Otherwise, there can be lots of different assignments of $p_i$ for $B_i$. Then, we have $q=(1-p)^{n-1}$ and $p \le 1-  c^{\frac{1}{n-1}}$. Let $c'=1-  c^{\frac{1}{n-1}}$. Based on Equation~\ref{eq:neg_p_pw}, we have:

\begin{equation}\label{eq:neg_f_pw}
\vspace{-1mm}
\begin{cases}
\frac{f-f^m}{1+f} \le c' & \text{if } m \text{ mod }2= 1 \\
\frac{f+f^m}{1+f} \le c' & \text{if } m \text{ mod }2= 0
\end{cases}
\vspace{-1mm}
\end{equation}

Since $n$ and $p_n$ are constants, $c'$ is a constant. Due to Abel-Ruffini theorem~\cite{Abel-theorem}, there are no direct formula to compute the value of $f$ for general $m$. Thus, we employ numeric solver to compute the range of $f$. After getting the feasible range of $f$, we and apply the same optimization strategy as in Sec.~\ref{sec:para_select_basic}, restricted to this range.

\vspace{-2mm}
\subsection{Tree Index for High-Cardinality Buckets}\label{sec:upper_tree}

Since $q=(1-p)^{n-1}$, the probability of an empty-answer query classified with $B_1$ in Lemma~\ref{lem:neg-misclass} and Lemma~\ref{lem:prefix-negative} is also affected by the number of buckets. Based on Equation~\ref{eq:neg_f_pw} and $c'=1-c^{\frac{1}{n-1}}$, where $c\in(0,1)$, a larger $n$ results in a smaller $c'$. Thus, a smaller $f$ is needed and leads to a larger Bloom filter size. We design a new strategy to reduce the number of buckets with the layered filters.

Specifically, we build a tree-based index to manage queries within $B_i$ if $i$ is larger than a predefined threshold. 
The nodes in the tree are stored in a contiguous array. Each node utilizes 4 bits to record the bucket ID (which bucket) of the query and 2 bytes to represent the index of its first child. The children of each node are stored sequentially, and a 1-bit flag is used to indicate whether a node is the last child in its sequence. This representation reduces memory overhead while preserving efficient query operations.

We incorporate the threshold value selection process into our parameter selection framework. Specifically, we first enumerate candidate thresholds (from $n$ to $1$) to determine which buckets will adopt the tree-based index. For each threshold, we then apply the same optimization strategy as in Sec.~\ref{sec:para_select_basic}.

\vspace{-1mm}
\section{Support Longer Queries}\label{sec:extend_query}
\tocheck{The core \BF{} approach supports prefix, suffix, and substring queries under a length constraint on the query string. Here, we extend \BF{} to handle longer query strings. For clarity, we first use ($\%S\%$) queries as an example to illustrate the main idea.}

Suppose the maximum substring length that \(\BF{}\) can accurately estimate is \(L\).
For a query string \(Q = q_1q_2\ldots q_{|Q|}\) with \(|Q|>L\), we estimate \(Card(\%Q\%)\) by converting the problem into a \emph{Markov process}~\cite{Markov} that models dependencies among consecutive substrings up to length \(L\). 
We rewrite the probability of \(Q\) (and hence its expected cardinality) using the chain rule:
$
P(q_1,\ldots,q_{|Q|}) = P(q_1,\ldots,q_L)
\prod_{i=L+1}^{|Q|}
P(q_i \mid q_1,\ldots,q_{i-1}).
$

Since \(\BF{}\) provides cardinality (selectivity) up to length \(L\),
we approximate each conditional term by conditioning on the most recent \(L-1\) characters,
leading to an \((L-1)\)-order Markov process:
$
P(q_i \mid q_1,\ldots,q_{i-1})
\approx
P(q_i \mid q_{i-L+1},\ldots,q_{i-1})
=
\frac{P(\%q_{i-L+1}\ldots q_i\%)}{P(\%q_{i-L+1}\ldots q_{i-1}\%)}$.
This yields:
$
\widehat{Card}(\%Q\%)
\approx
|\mathcal{S}|\,
P(q_1,\ldots,q_L)
\prod_{i=L+1}^{|Q|}
\frac{P(\%q_{i-L+1}\ldots q_i\%)}{P(\%q_{i-L+1}\ldots q_{i-1}\%)}
$, where \(|\mathcal{S}|\) denotes the total number of strings. This approximation assumes that the probability of a new character depends primarily on its most recent \(L-1\) predecessors, i.e., long-range correlations are negligible once the last \(L\) characters are known.
Empirically, most real string distributions exhibit such local dependency~\cite{MOKVI}.

\begin{table}[t]
\centering
\caption{Statistics of string sets ($\mathcal{S}_D$) and query sets ($\mathcal{S}_Q$).}
\vspace{-2mm}
\begin{tabular}{|c|c|r|r r|r|r r|}
\hline
\multirow{2}{*}{Dataset} & \multirow{2}{*}{$|\mathcal{V}|$} & \multicolumn{3}{c|}{String Set ($\mathcal{S}_D$)} & \multicolumn{3}{c|}{Query Set ($\mathcal{S}_Q$)} \\
\cline{3-8}
 &  & $|\mathcal{S}_D|$ & $\ell_{\text{avg}}$ & $\ell_{\text{max}}$ & $|\mathcal{S}_Q|$ & $\ell_{\text{avg}}$ & $\ell_{\text{max}}$ \\
\hline
Author   & 53 & 111,162 & 13.86  & 43 & 414,614 & 7.70 & 10 \\
DBLP-AN   & 27    & 450,000 & 14.61  & 44 & 701,460 & 8.26 & 10 \\
IMDB-AN   & 27    &   550,000  & 14.14 & 39 & 750,000  & 8.47 & 10 \\
IMDB-MT   & 38     &   357,923  & 18.44 & 79 & 750,000  & 8.19 & 10 \\
\hline
\end{tabular}
\label{tab:dataset-stats}
\end{table}


\smallskip\noindent\textbf{Theoretical Error Bound (under an $(L\!-\!1)$-order Markov assumption).}
Let $e_b$ be the maximum Q-error guarantee of \BF{} for any substring
of length $\le L$: $e_b^{-1}C(U)\le \widehat C(U)\le e_b\,C(U) (\forall\, |U|\le L).$
Under this assumption, for any query $Q$ with $|Q|>L$, the Q-error of the
Markov-based estimator is bounded by
$e_b^{\,1+2(|Q|-L)}.$
The bound follows because the estimator uses one $L$-gram term
(with error $\le e_b$) and $|Q|-L$ ratio terms, each involving two
$\le L$ substrings (contributing at most $e_b^2$ per step).

\smallskip\noindent\textbf{Discussion.} The same Markov-process idea can apply to prefix queries. For $|Q|>L$, we adopt an  $(L-1)$-order Markov model to express: $
\widehat{Card}(\%Q\%) \;\approx\; |\mathcal{S}|\; P(q_1\ldots q_L\%)
\prod_{i=L+1}^{|Q|}
\frac{P(\%q_{i-L+1}\ldots q_i\%)}{P(\%q_{i-L+1}\ldots q_{i-1}\%)}
$. By reversing the query string and treating the suffix as a prefix, the $(L-1)$-order Markov estimator can support suffix queries as well.

\section{Experimental Study}\label{sec:exp}
\subsection{Experimental Setup}\label{sec:exp-setup}
\noindent\textbf{Datasets}. For synthetic testing, where queries are randomly generated, we conduct experiments on four widely used datasets, referenced in existing studies~\cite{LPLM,Astrid,E2E,FACE}. The detailed statistics of these datasets are presented in Table~\ref{tab:dataset-stats}\footnote{We exclude \textsf{PartName} in TPC-H~\cite{TPCH}, which is widely used in existing benchmarks, as it is synthetically generated and contain repeated words, which do not reflect real-world query patterns.}.



\begin{table*}[ht]
\centering
\caption{Q-error for different estimators on valid queries}
\vspace{-2mm}
\label{tab:qerror_basic_patterns}
\resizebox{0.99\textwidth}{!}{
\begin{tabular}{|c|c|c|c|c|c|c|c|c|c|c|c|c|c|c|c|c|c|}
\hline
\multirow{2}{*}{Query} & \multirow{2}{*}{Method}
  & \multicolumn{4}{c|}{Author} 
  & \multicolumn{4}{c|}{IMDB-AN} & \multicolumn{4}{c|}{DBLP-AN} & \multicolumn{4}{c|}{IMDB-MT}\\\cline{3-18}
 & &Mean & 90th & 99th & Max 
     & Mean & 90th & 99th & Max & Mean & 90th & 99th & Max 
     & Mean & 90th & 99th & Max \\\hline
\multirow{6}{*}{S\%}
  & PostgreSQL      &  9.6   & 11.0 & 11.0  & 1098.0   
                   &  9.81  & 11.0 & 11.0  & 838.09   
                   &  7.99  & 11.0  &  11.0\second   & 1355.18  
                   &  9.36  & 11.0  &  11.0   & 1098.0   \\
                   
  & \BF{}    &  1.38\best&1.41&1.46\best&1.5\best 
                   &  1.4\best&1.41&1.41\best&1.5\best
                    &  1.37\best&1.41\best&1.46\best&1.5\best  
                   &  1.39\best&1.41&1.46\best&1.5\best \\
 
  & E2E             &  2.30  & 1.49  &  12.26 & 1446.39 
                   &  1.74   & 1.01\second &  10.73 & 370.95
                   &  2.57  & 1.82  & 15.26 & 474.23\second   
                   &  1.99  & 1.10\second &  12.43 & 4727.97 \\
                   
& Astrid          &  1.94&1.24\best&10.68&412.31  
                   &  1.61&1.0\best&9.5&876.22 
                   &  2.32&1.67&12.76&768.65  
                   &  1.78&1.03\best&11.78&211.02\second   \\
                   
  & LPLM            &  8.98  & 3.41  &  86.53 & 7269.58 
                   &  4.93   & 2.73  &  34.65 & 10829.71 
                   &  8.95  & 3.11  &  102.48& 5303.01 
                   &  6.55  & 2.21  &  65.66 & 28987.37 \\
                   
  & CLIQUE          &  1.7\second & 1.35\second  & 6.33\second & 166.0\second 
                   &  1.5\second  & 1.07  & 7.24\second & 292.02\second  
                   &  2.19\second&1.62\second&11.4&589.77 
                   &  1.66\second&1.23&8.29\second&215.5 \\
\hline\hline
\multirow{6}{*}{\%S}
  & PostgreSQL      &  9.5   & 11.0 &  11.0  & 1098.0   
                   &  9.83&11.0&11.0&7220.64   
                   &  7.91  & 11.0  &  11.0   & 2257.18  
                   &  9.46  & 11.0  &  11.0   & 2226.09  \\
                   
  & \BF{}    &  1.38\best&1.41\second&1.46\best&1.5\best  
                   &  1.4\best&1.41&1.46\best&1.5\best 
                &  1.37\best&1.41\best&1.46\best&1.5\best  
                   &  1.39\best&1.41&1.46\best&1.5\best  \\
 
  & E2E             &  2.08  & 1.52  & 9.03  & 589.68  
                   &  1.99   & 1.06\second     & 12.88 & 12785.0
                    &  2.44  & 1.83  &  12.62 & 842.46  
                   &  1.95  & 1.12\second &  11.00 & 1147.66  \\
                   
 & Astrid          &  2.63&1.3\best&17.49&3418.58  
                   &  2.32&1.0\best&19.98&4606.3
                    &  3.02&1.62\second&22.49&2833.54 
                   &  2.07&1.01\best&16.88&1731.33  \\
                   
  & LPLM            &  4.75  & 3.42  & 22.05 & 4727.97 
                   &  7.57&3.21&69.68&10798.68
                   &  6.25  & 3.14  &  33.73 & 23830.59
                   & 11.97  & 3.00  &  83.22 & 87207.32 \\
                   
  & CLIQUE          &  1.77\second&1.44&6.45\second&50.8\second  
                   &  1.59\second&1.1&7.82\second&927.21\second
                   &  2.23\second&1.75&10.84\second&267.64\second  
                   &  1.62\second&1.2&8.0\second&356.34\second \\
                   
\hline\hline
\multirow{9}{*}{\%S\%}
  & PostgreSQL      &  9.17  & 11.0 &  11.0  & 1098.0   
                   &  9.50  & 11.0 & 11.0  & 5213.55
                   &  7.74  & 11.0  &   11.0\second  & 1547.0   
                   &  9.33  & 11.0  & 11.0   & 7752.73  \\
                   
  & MO             & 2.88&1.0\best&11.01&16056.0
                    & 1.51&1.0\best&6.0\second&19600.0
                    & 2.33&1.02\second&14.0&1321.31
                    & 1.5\second&1.0\best&5.6\second&4542.0  \\
                    
  & LBS             & 2.15&1.0\best&13.0&39.4\second
                    & 1.52&1.0\best&9.0&117.07
                    & 2.47&1.0\best&15.26&277.15
                    & 1.62&1.0\best&10.0&90.91  \\
    
  & \BF{}    &  1.38\best&1.41\second&1.46\best&1.5\best  
                   &  1.4\best&1.41&1.46\best&1.5\best
                    &  1.37\best&1.41&1.46\best&1.5\best  
                   &  1.39\best&1.41&1.46\best&1.5\best  \\
 
  & E2E             &  2.37  & 1.57  &  11.73 & 1278.92 
                   &  2.14&1.16&16.15&2901.0 
                   &  2.72  & 1.88  &  17.02 & 847.46  
                   &  2.27  & 1.27 & 14.39 & 3961.14  \\
                   
 & Astrid          &  2.92&1.42&20.46&6879.51 
                   &  2.83&1.0\best&24.89&8016.84
                   &  3.37&1.68&27.9&6385.91 
                   &   2.68&1.01\second&23.96&6345.04      \\
                   
  & LPLM            & 10.32  & 4.19  &  113.9 & 9451.88 
                   & 9.93&3.84&92.56&15360.61
                    &114.86  & 4.05  &  190.55&17860.03
                   & 18.36  & 3.88  &  220.79&45274.68 \\

  & CLIQUE          &  1.83\second&1.46&6.89\second&166.0  
                   &  1.62&1.13\second&7.91&557.28
                    &  2.3\second&1.72&11.93&479.0  
                   &  1.7&1.24&8.81&908.69 \\
                   
  & SSCard        & 2.02&1.0\best&12.0&59.0
                    & 1.45\second&1.0\best&7.0&72.0\second
                     & 2.41&1.0\best&16.0&82.0\second
                    & 1.58&1.0\best&9.0&71.0\second  \\\hline

                   \multicolumn{18}{@{}l}{\small Note: \best\ marks the best (lowest) value, while \second\ marks the second-best.}
\end{tabular}\label{tab:basic-valid-acc}}
\vspace{-4mm}
\end{table*}

\noindent\textbf{Methods for Comparison}. We compare the following methods:
\begin{itemize}[leftmargin=*]
    \item PostgreSQL~\cite{PG}: We utilize the \textsf{EXPLAIN} statement to obtain the cardinality estimation of \textsf{LIKE} queries.
    \item MO~\cite{MOKVI}: A substring cardinality estimator based on
pruned suffix tree. We set the prune threshold to 2\%, reserving 2\% nodes of the full suffix tree.
    \item LBS~\cite{LBS}: One of the traditional state-of-the-art methods designed for patterns of the form $\%S\%$ based on $n$-gram.
    \item E2E~\cite{E2E}: A deep learning-based method that uses sample bitmap as encoding and supports $S\%$, $\%S$, and $\%S\%$.
    \item Astrid~\cite{Astrid}: A state-of-the-art deep learning method that supports $S\%$, $\%S$, and $\%S\%$. We adopt its embedding-based variant, as it performs better than the natural language model-based variant.
    \item LPLM~\cite{LPLM}: A state-of-the-art learning-based method designed for universal \textsf{LIKE} queries.
    \item CLIQUE~\cite{CLIQUE}: A state-of-the art learning-based method designed to support basic patterns and combined patterns.
    \item \BF{}: Our proposed estimator that focuses on $S\%$, $\%S$, and $\%S\%$. The source code is available at~\cite{TR}.
\end{itemize}
We adopt PostgreSQL v$13.1$ and for other baselines we adopt the codes provided by the authors and use the default settings. For our method, we set $eb=1.5$ in default, enabling the frontier-based optimization (cf. Sec.~\ref{sec:frontier}) and building the tree structure for the large cardinality (cf. Sec.~\ref{sec:upper_tree}). Note that for each method, we compare it only against the query patterns that it supports. 

\noindent\textbf{Query Generation}. 
We consider three \textsf{LIKE} query patterns: $S\%$, $\%S$, $\%S\%$, where $|S| <= 10$. For non-empty-answer queries, we evaluate them on the same $250,000$ queries\footnote{We aim to maintain a training-to-testing ratio of 2:1. If it is not possible to sample the desired number of queries, we adjust the sample sizes while preserving this ratio.} by random sampling. 

Existing learning-based methods often adopt their own training and testing data generation strategies, varying in query patterns and the number of queries used. 
To ensure a fair and meaningful comparison, we train all learning-based methods on the same training queries. This unified setup ensures that performance differences reflect model design rather than discrepancies in data exposure.

\vspace{-1mm}
Specifically, we randomly sample $500,000$ queries for each query pattern as training data . For methods that train a single model across all query patterns -- such as LPLM, E2E, and CLIQUE -- we combine the training and testing data from all query patterns to train and evaluate a unified model. For empty-answer query testing, we generate queries by appending a random number of characters to existing non-empty-answer testing queries. 

\noindent\textbf{Evaluation Aspects}. We evaluate all the methods above from:

\noindent\underline{\textit{Accuracy}} -- How closely do their predictions match the ground truth for fundamental query patterns in \textsf{CE4Str}? We use Q-error as the evaluation metric, reporting both the mean Q-error and its distribution ($50\%$, $90\%$, $99\%$, and $100\%$ quantile).

\noindent\underline{\textit{Robustness}} -- Do these methods maintain consistent accuracy across different cardinality ranges and the lengths of queries? We break down accuracy across varying cardinality ranges and queries' lengths. We report the corresponding Q-error for each range/length. 

\noindent\underline{\textit{Empty-answer Queries Support}} -- How well these methods support empty-answer queries? We report the percentage of empty-answer queries with the estimated cardinality smaller than $2$.

\noindent\underline{\textit{Overhead}} -- How quickly can each method start to handle online queries (preparation time), and how efficiently can it perform online estimation (inference time)? Evaluating both metrics helps ensure that the method is not only accurate but also practical for deployment. We also report the on-disk size (storage usage) of the data structures used by each method.

\noindent\underline{\textit{Tunability \& Optimization}} -- How does \BF{} perform under different user-specified parameters, and what are the benefits of our proposed optimizations?





\vspace{-3mm}
\subsection{Overall Estimation Accuracy}\label{sec:exp-acc}
Table~\ref{tab:basic-valid-acc} reports the Q-error of six estimators across four datasets. We summarize the key observations as follows: (1) Overall, \BF{} consistently ranks among the top two methods across nearly all datasets and Q-error percentiles, demonstrating strong robustness. For instance, on \textsf{DBLP-AN}, \BF{} achieves a mean Q-error of 1.37 on all query patterns and significantly outperforms the second-best, CLIQUE with a mean Q-error around $2.2$. 
(2) In terms of worst-case behavior, \BF{} shows strong tail robustness. Across all datasets, its maximum Q-error remains under 1.5, whereas methods like CLIQUE and LPLM suffer from high variance -- e.g., on \textsf{DBLP-AN} with the substring query, CLIQUE reaches a maximum Q-error of 479.0, and LPLM exceeds 1000. 
(3) PostgreSQL consistently underperforms across all metrics. Its mean Q-error ranges from 7.9 to 9.6. While learning-based methods like LPLM and E2E improve over PostgreSQL in the mean case, they suffer from high variance and large tail errors. (4) On $\%S\%$, methods specifically designed for substring patterns (MO, LBS, SSCard) achieve impressive accuracy, particularly in terms of 90th percentile Q-error.

\begin{figure*}[t]
 \vspace{-3mm}
  \centering
  \begin{subfigure}[b]{\textwidth}
    \centering
    \includegraphics[width=0.7\textwidth]{ 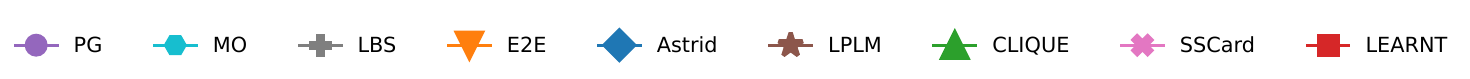}
  \end{subfigure}\\[-2mm]
    \begin{subfigure}[b]{0.25\textwidth}
    \centering
    \includegraphics[width=\textwidth]{ 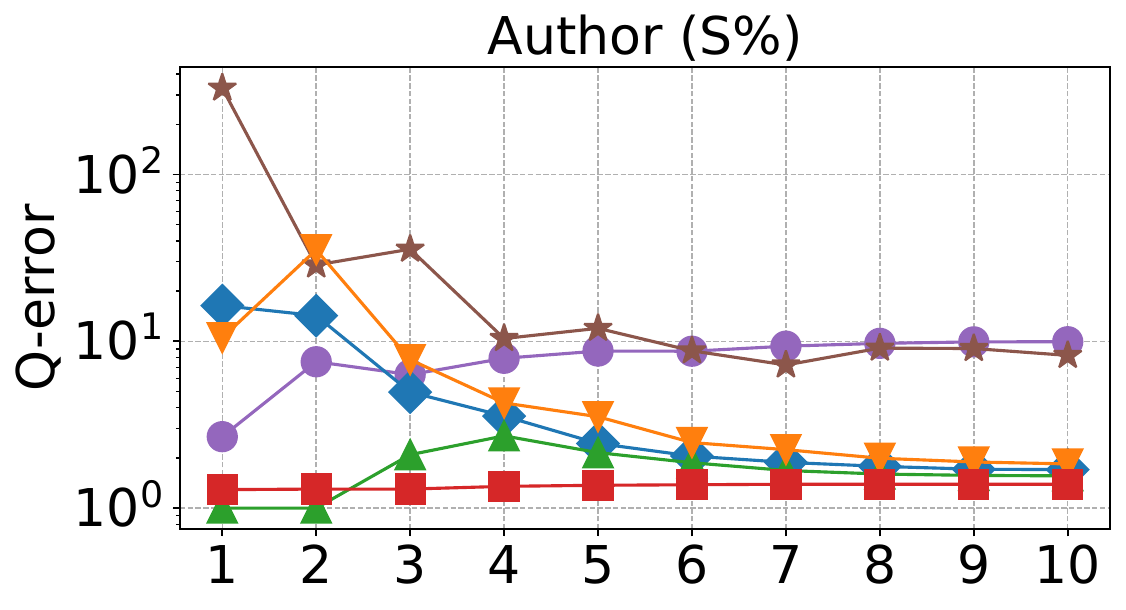}
    \label{fig:author_diff_len_prefix}
  \end{subfigure}%
  \hfill
  \begin{subfigure}[b]{0.24\textwidth}
    \centering
    \includegraphics[width=\textwidth]{ 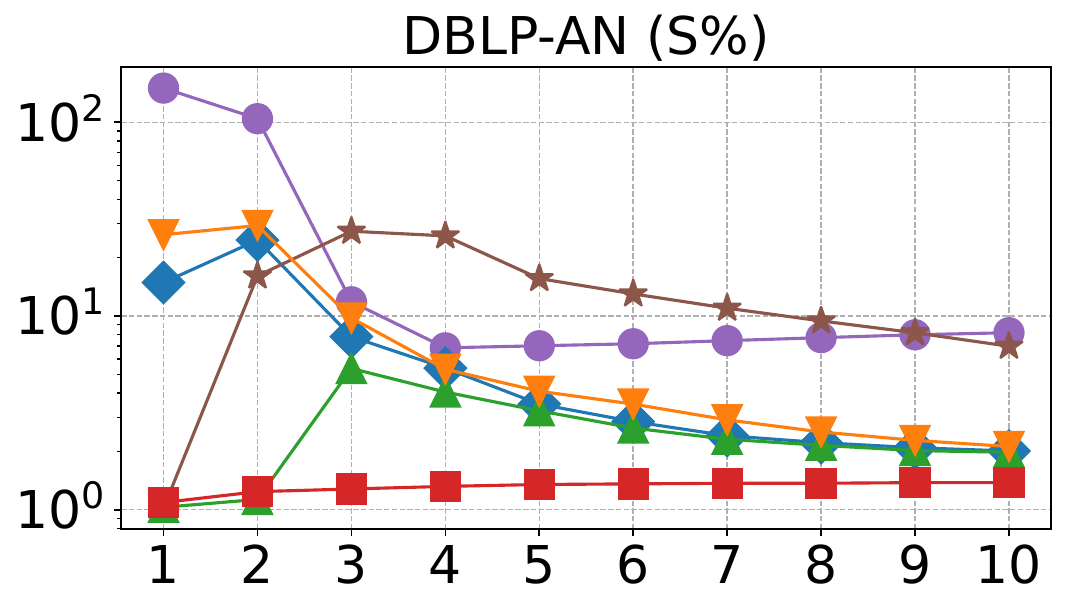}
    \label{fig:dblp_an_diff_len_prefix}
  \end{subfigure}%
  \hfill
  \begin{subfigure}[b]{0.24\textwidth}
    \centering
    \includegraphics[width=\textwidth]{ 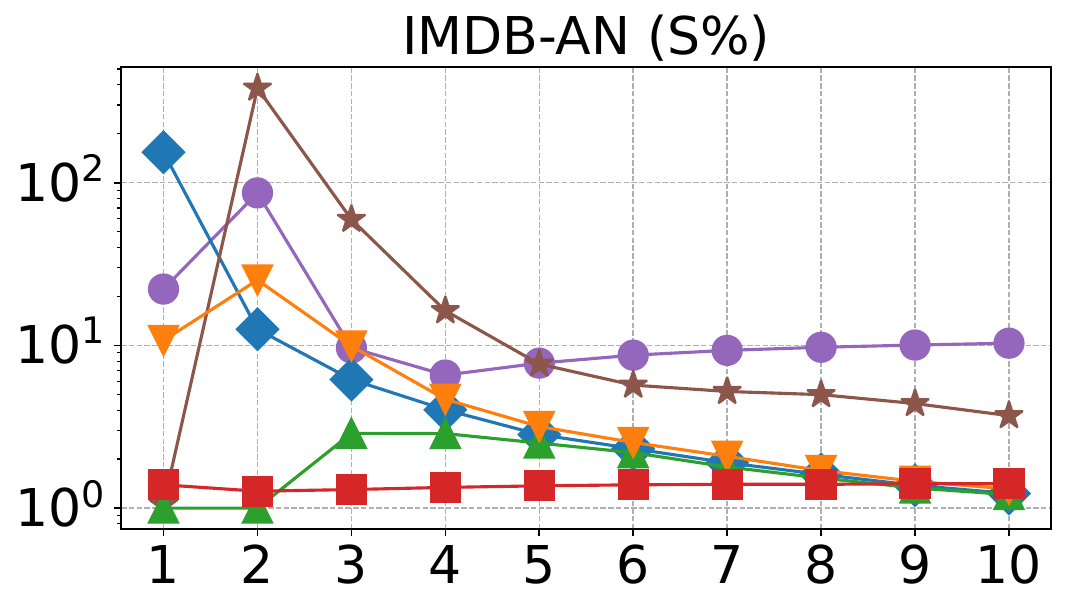}
    \label{fig:imdb_an_diff_len_prefix}
  \end{subfigure}%
  \hfill
  \begin{subfigure}[b]{0.24\textwidth}
    \centering
    \includegraphics[width=\textwidth]{ 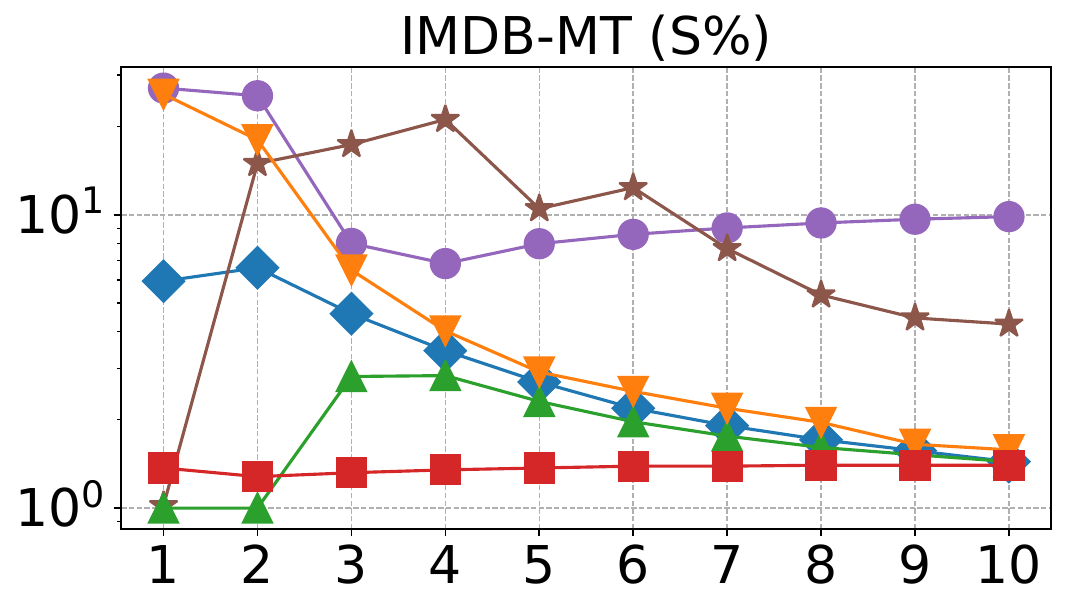}
    \label{fig:imdb_mt_diff_len_prefix}
  \end{subfigure}\\[-2mm]
  \begin{subfigure}[b]{0.25\textwidth}
    \centering
    \includegraphics[width=\textwidth]{ 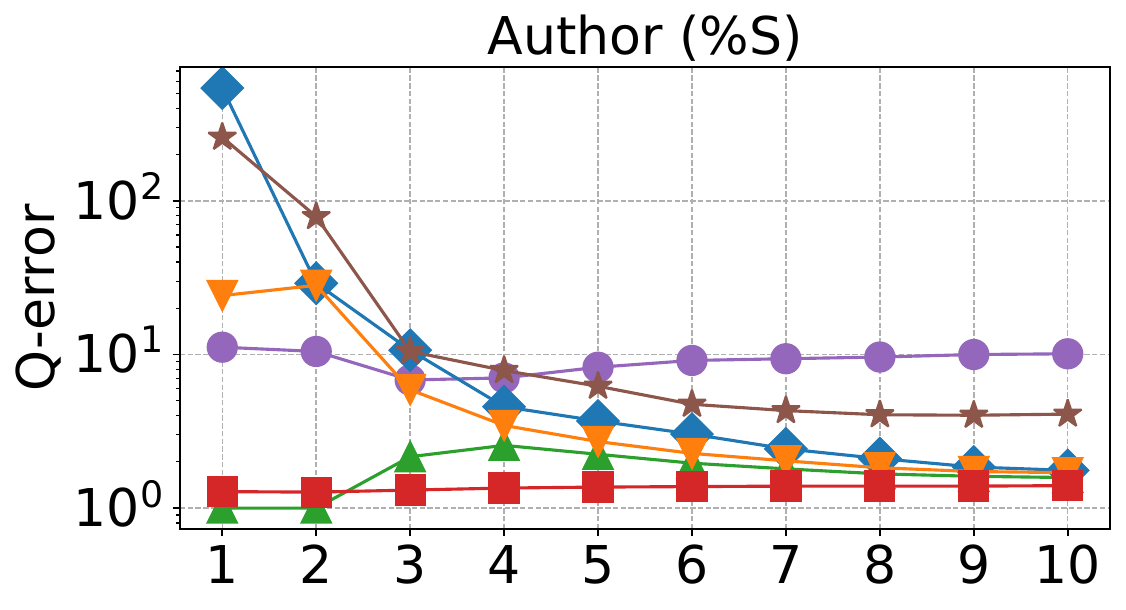}
    \label{fig:author_diff_len_suffix}
  \end{subfigure}%
  \hfill
  \begin{subfigure}[b]{0.24\textwidth}
    \centering
    \includegraphics[width=\textwidth]{ 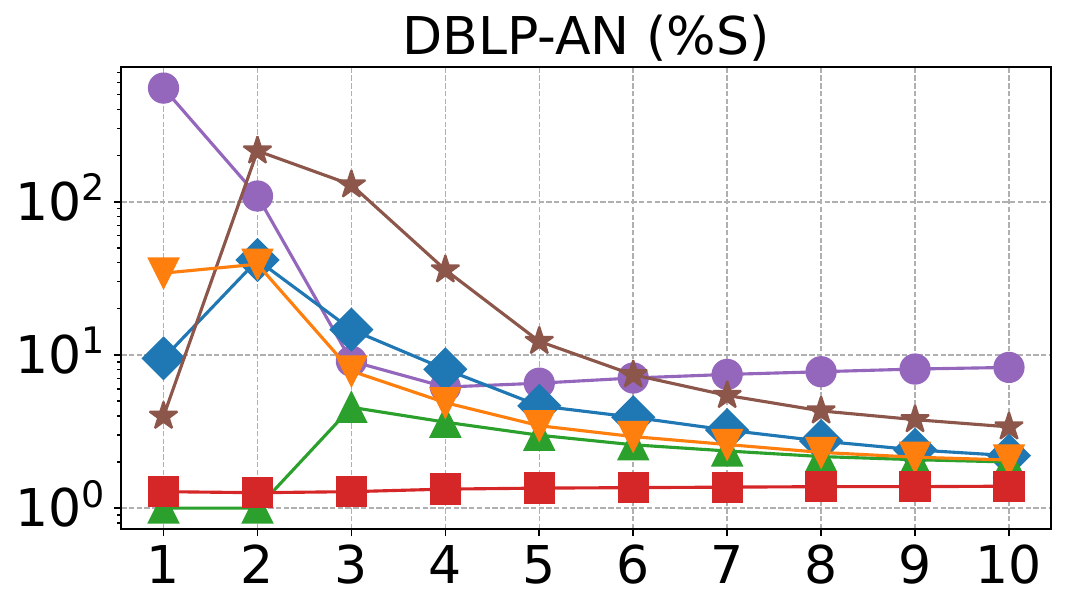}
    \label{fig:dblp_an_diff_len_suffix}
  \end{subfigure}%
  \hfill
  \begin{subfigure}[b]{0.24\textwidth}
    \centering
    \includegraphics[width=\textwidth]{ 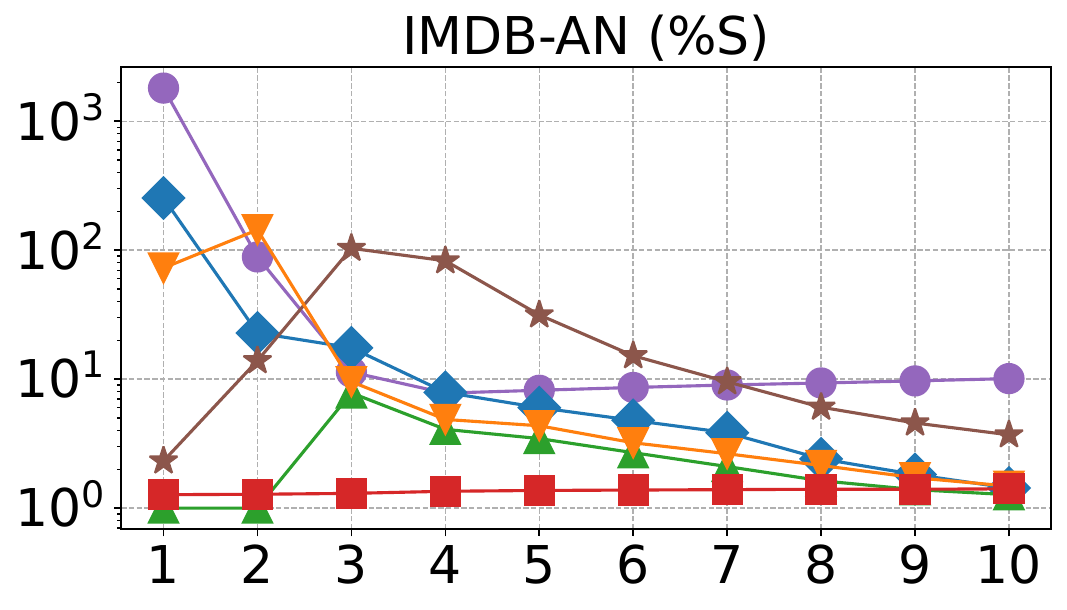}
    \label{fig:imdb_an_diff_len_suffix}
  \end{subfigure}%
  \hfill
  \begin{subfigure}[b]{0.24\textwidth}
    \centering
    \includegraphics[width=\textwidth]{ 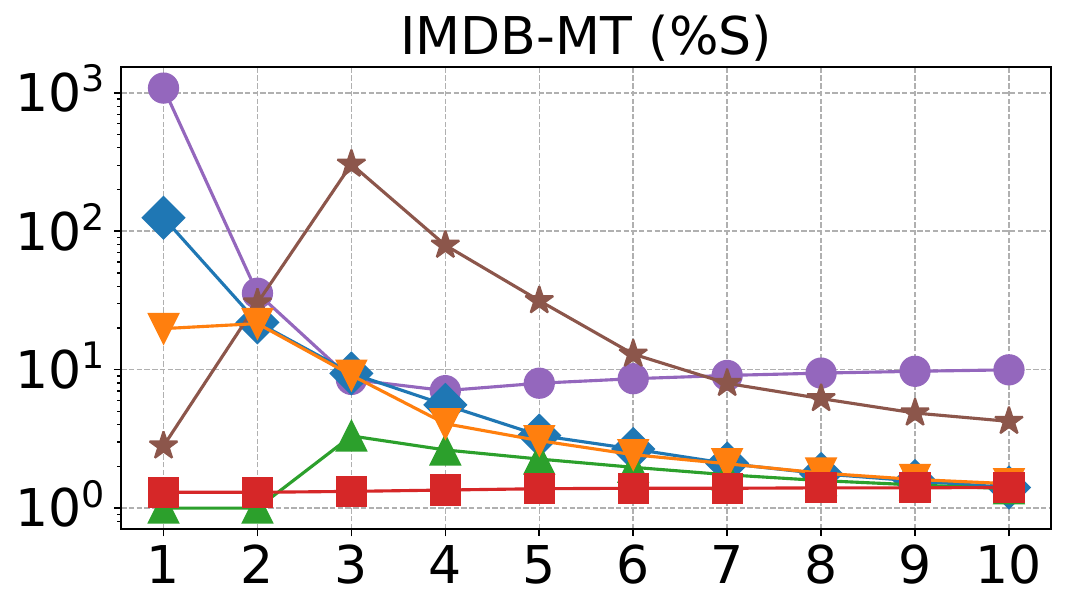}
    \label{fig:imdb_mt_diff_len_suffix}
  \end{subfigure}\\[-3mm]
  \begin{subfigure}[b]{0.25\textwidth}
    \centering
    \includegraphics[width=\textwidth]{ 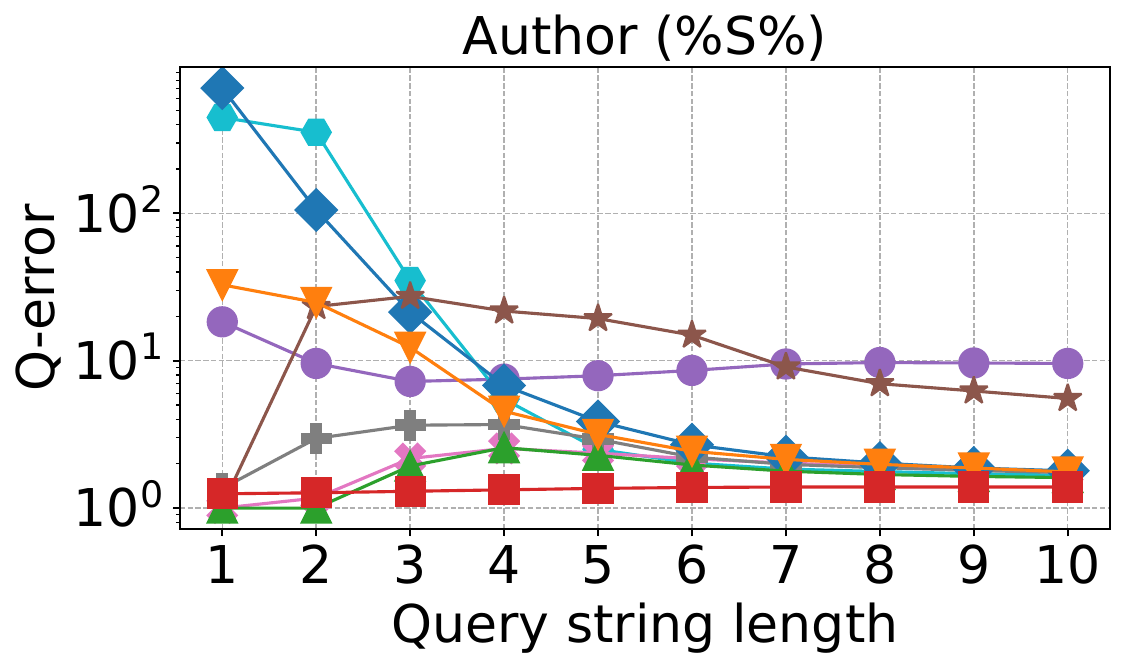}
    \label{fig:author_diff_len_substring}
  \end{subfigure}%
  \hfill
  \begin{subfigure}[b]{0.24\textwidth}
    \centering
    \includegraphics[width=\textwidth]{ 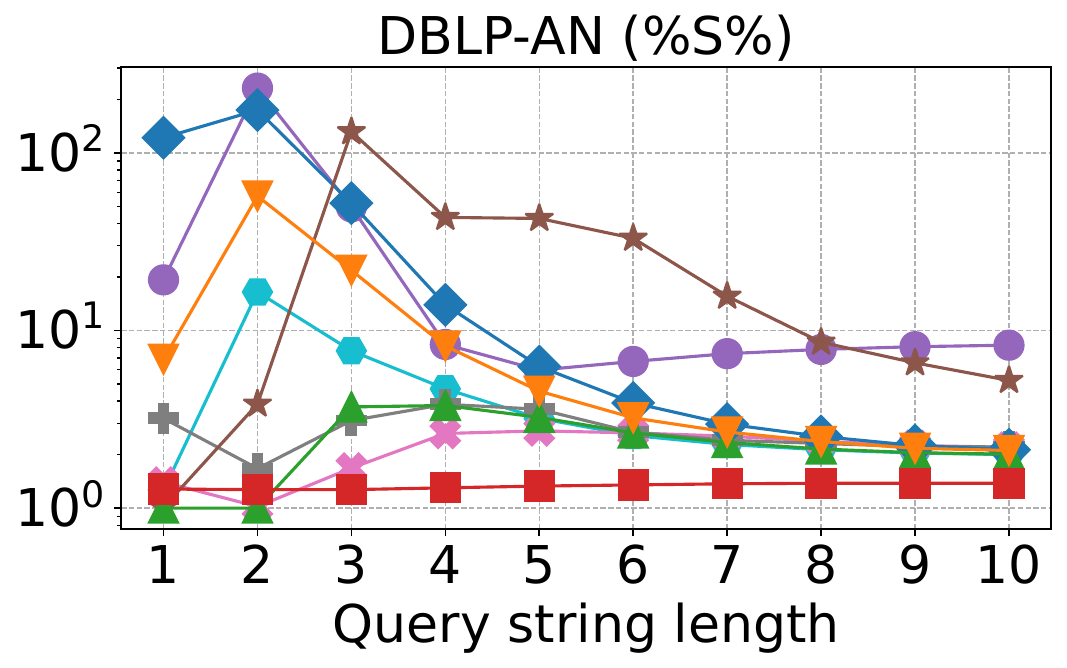}
    \label{fig:dblp_an_diff_len_substring}
  \end{subfigure}%
  \hfill
  \begin{subfigure}[b]{0.24\textwidth}
    \centering
    \includegraphics[width=\textwidth]{ 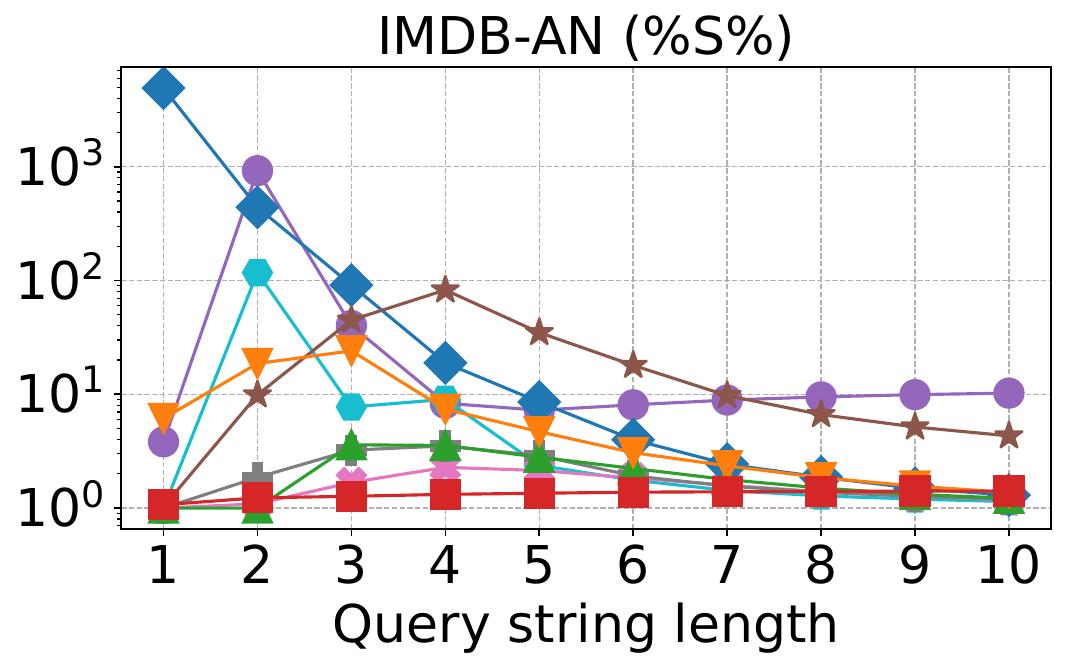}
    \label{fig:imdb_an_diff_len_substring}
  \end{subfigure}%
  \hfill
  \begin{subfigure}[b]{0.24\textwidth}
    \centering
    \includegraphics[width=\textwidth]{ 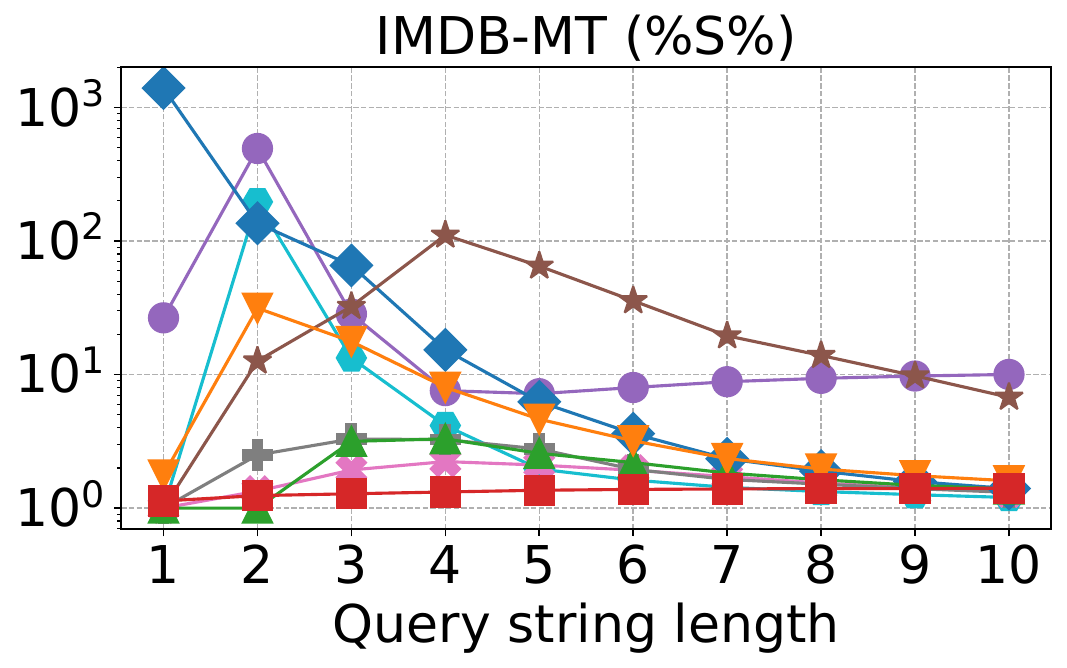}
    \label{fig:imdb_mt_diff_len_substring}
  \end{subfigure}
    \vspace{-6mm}
  \caption{Q-error of all estimators as pattern string length varies.}
  \vspace{-4mm}
  \label{fig:diff_lens}
\end{figure*}

\vspace{-4mm}
\subsection{Robustness of Each Estimator}\label{sec:exp-basic-robust}
\noindent\textbf{Robustness on Different Query Lengths}. Figure~\ref{fig:diff_lens}  illustrates how Q-error changes with query string length across four benchmarks and three query patterns ($S\%$, $\%S$, $\%S\%$). We observe the following trends: (1) Our method consistently yields the lowest Q-error across almost all lengths and query patterns, showcasing strong robustness and generalization. It usually delivers a $1.5-2$x lower Q-error than the next-best method. For example, when the query length $\in [3,7]$, it maintains Q-error close to 1, while others -- especially LPLM and E2E -- can exceed 5; (2) As query strings grow longer, most methods improve due to increased selectivity and they tends to predict the cardinality with small value; (3) When the query length is 1 or 2, CLIQUE achieves slightly better best Q-error than ours. CLIQUE introduces extend N-gram table. to store all frequent N-grams. The queries with small length, i.e., 1 or 2, usually have a large cardinality and are included in the extend N-gram table. Thus, CLIQUE directly answers the query with the exact cardinality without calling the learning model. (4) Learned estimators like LPLM and E2E show high variance and unstable behavior at small lengths (especially length = 1–3), sometimes spiking above $10^2-10^3$.

\begin{figure*}[!h]
  \centering
  \begin{subfigure}[b]{\textwidth}
    \centering
    \includegraphics[width=0.5\textwidth]{ 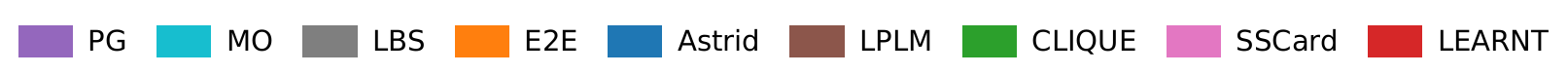}
  \end{subfigure}
\vspace{-5mm}
    \begin{subfigure}[b]{0.25\textwidth}
    \centering
    \includegraphics[width=\textwidth]{ 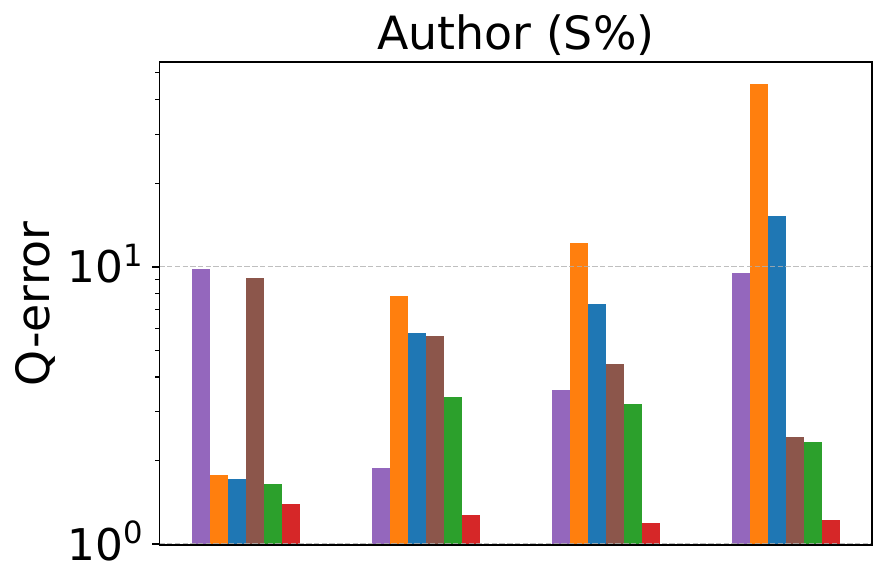}
    \label{fig:author_diff_range_prefix}
  \end{subfigure}%
  \hfill
  \begin{subfigure}[b]{0.24\textwidth}
    \centering
    \includegraphics[width=\textwidth]{ 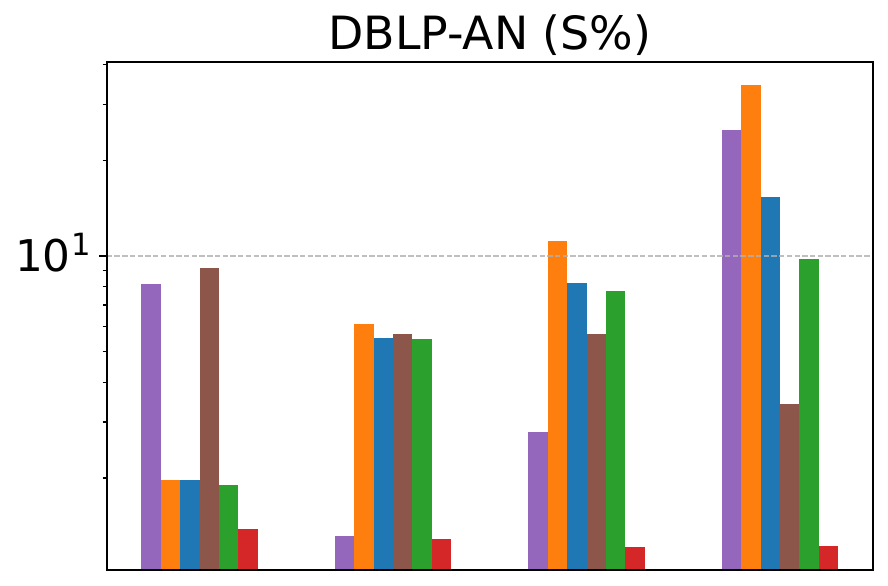}
    \label{fig:dblp_an_diff_range_prefix}
  \end{subfigure}%
  \hfill
  \begin{subfigure}[b]{0.24\textwidth}
    \centering
    \includegraphics[width=\textwidth]{ 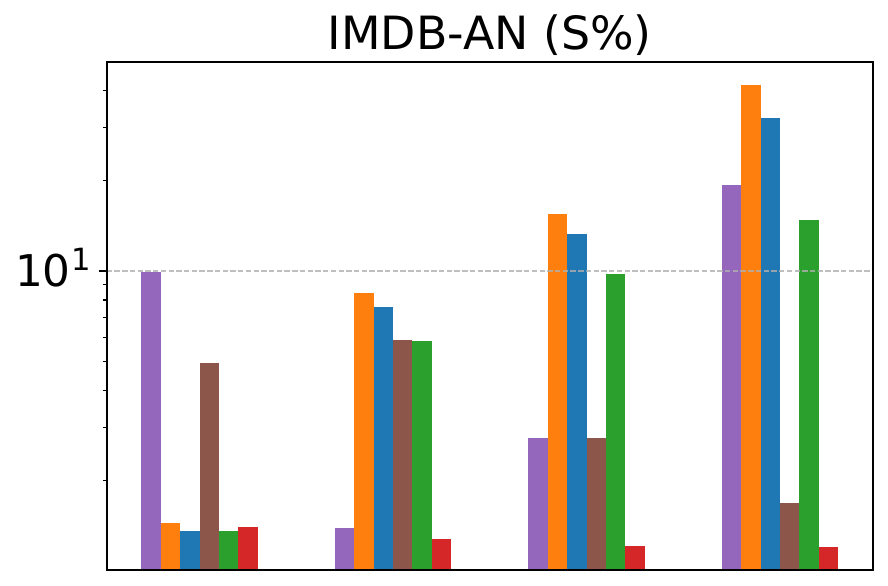}
    \label{fig:imdb_an_diff_range_prefix}
  \end{subfigure}%
  \hfill
  \begin{subfigure}[b]{0.24\textwidth}
    \centering
    \includegraphics[width=\textwidth]{ 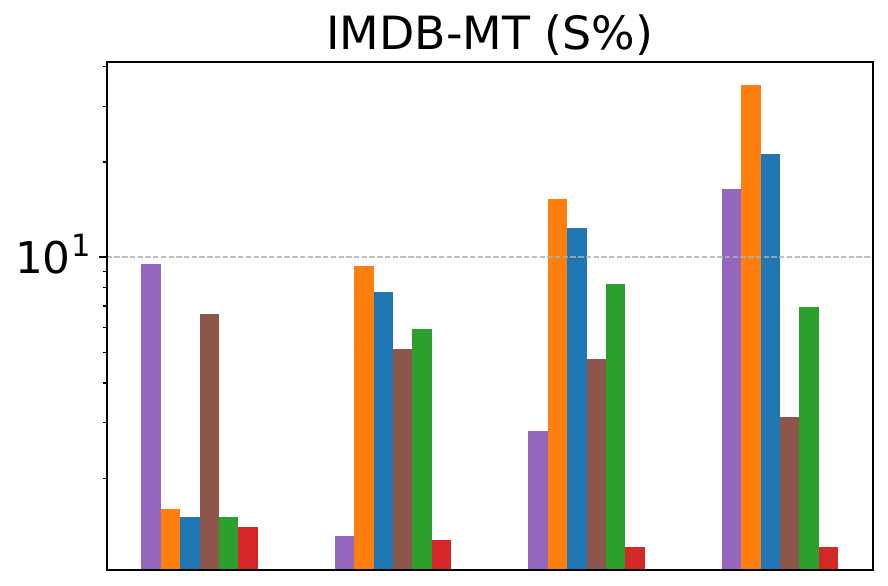}
    \label{fig:imdb_mt_diff_range_prefix}
  \end{subfigure}\\[-0.5mm]
  
  \begin{subfigure}[b]{0.25\textwidth}
    \centering
    \includegraphics[width=\textwidth]{ 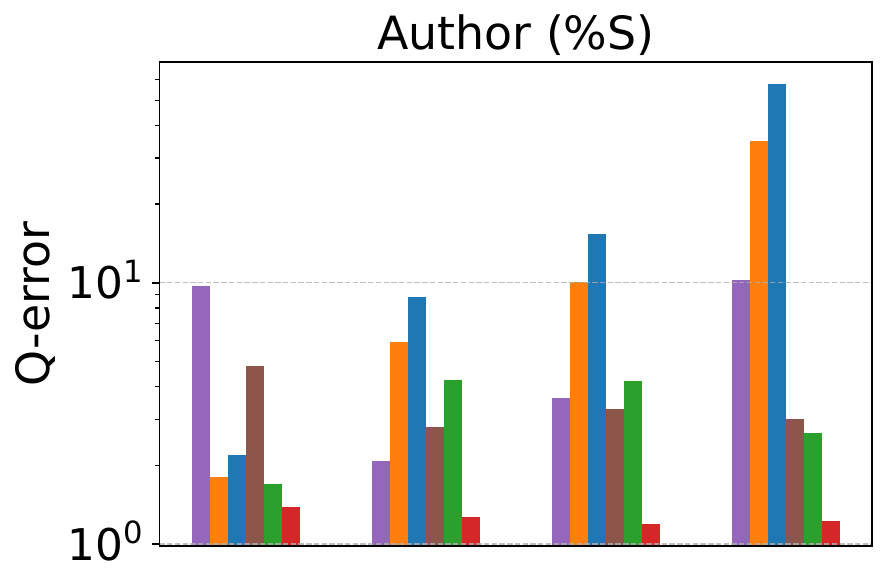}
    \label{fig:author_diff_range_suffix}
  \end{subfigure}%
  \hfill
  \begin{subfigure}[b]{0.24\textwidth}
    \centering
    \includegraphics[width=\textwidth]{ 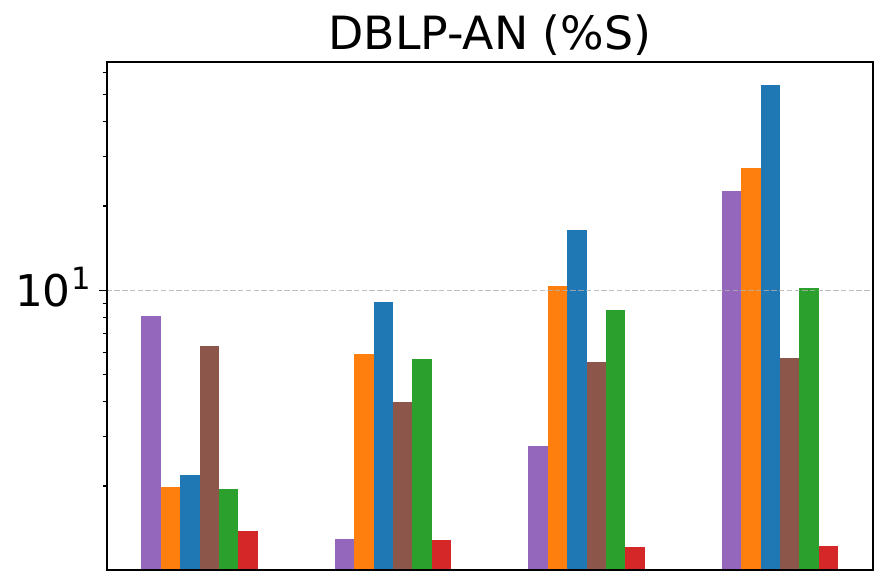}
    \label{fig:dblp_an_diff_range_suffix}
  \end{subfigure}%
  \hfill
  \begin{subfigure}[b]{0.24\textwidth}
    \centering
    \includegraphics[width=\textwidth]{ 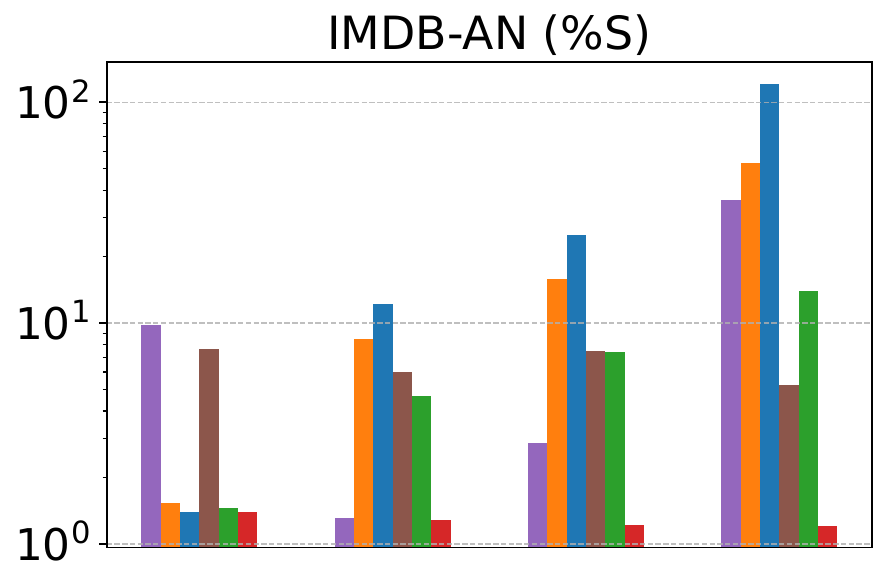}
    \label{fig:imdb_an_diff_range_suffix}
  \end{subfigure}%
  \hfill
  \begin{subfigure}[b]{0.24\textwidth}
    \centering
    \includegraphics[width=\textwidth]{ 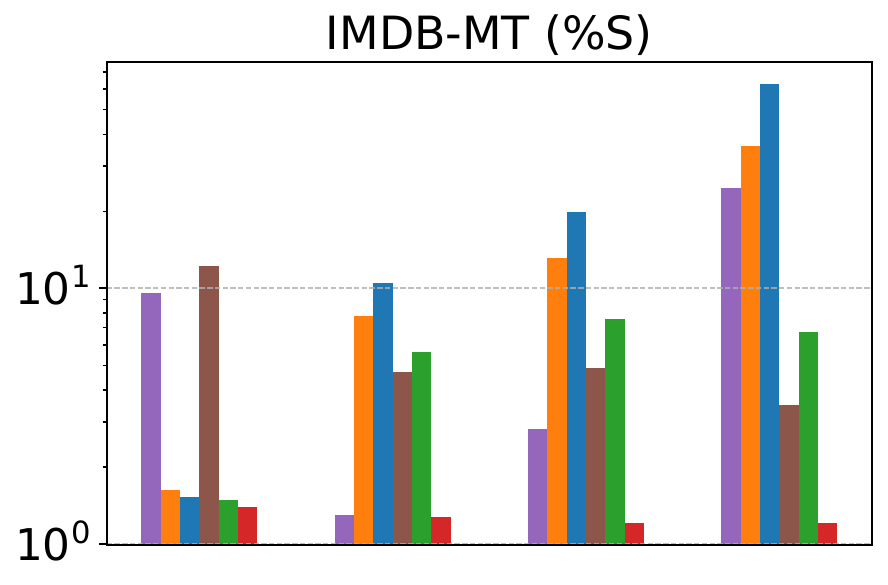}
    \label{fig:imdb_mt_diff_range_suffix}
  \end{subfigure}\\[-3mm]
  \begin{subfigure}[b]{0.25\textwidth}
    \centering
    \includegraphics[width=\textwidth]{ 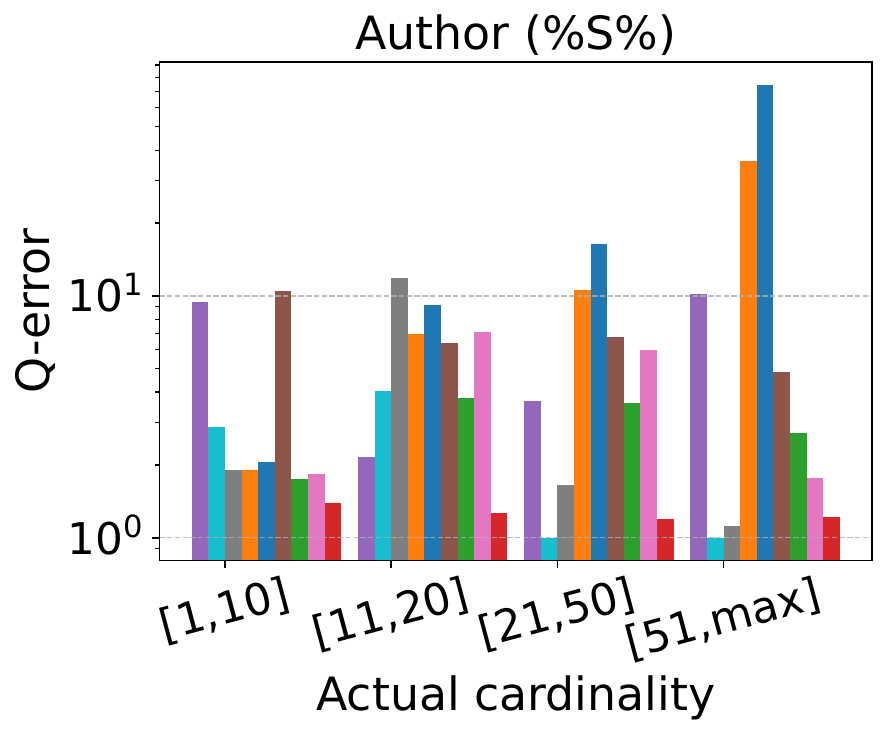}
    \label{fig:author_diff_range_sub}
  \end{subfigure}%
  \hfill
  \begin{subfigure}[b]{0.24\textwidth}
    \centering
    \includegraphics[width=\textwidth]{ 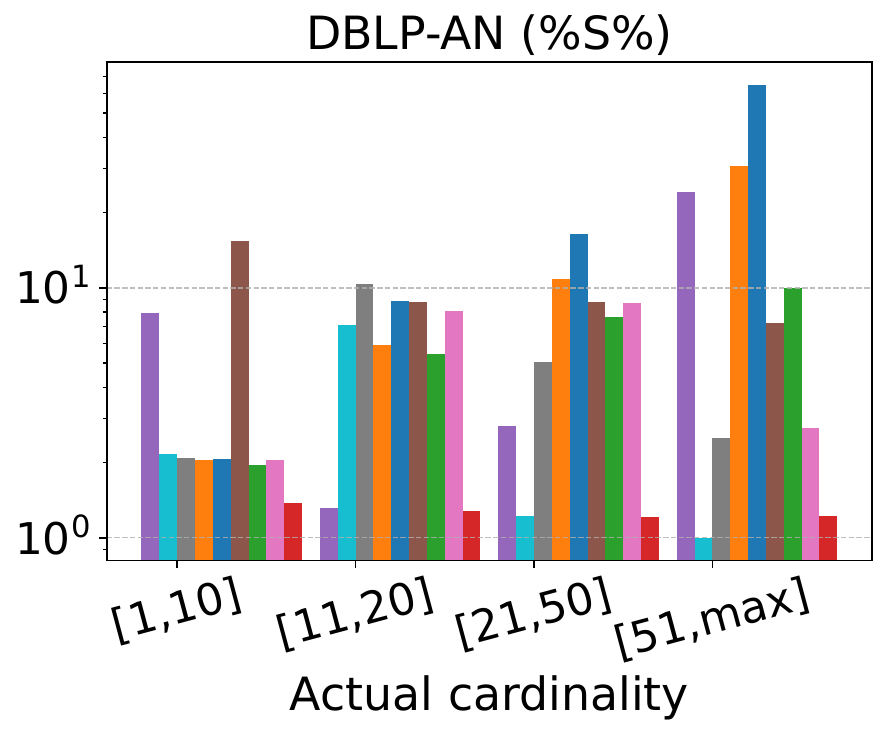}
    \label{fig:dblp_an_diff_range_sub}
  \end{subfigure}%
  \hfill
  \begin{subfigure}[b]{0.24\textwidth}
    \centering
    \includegraphics[width=\textwidth]{ 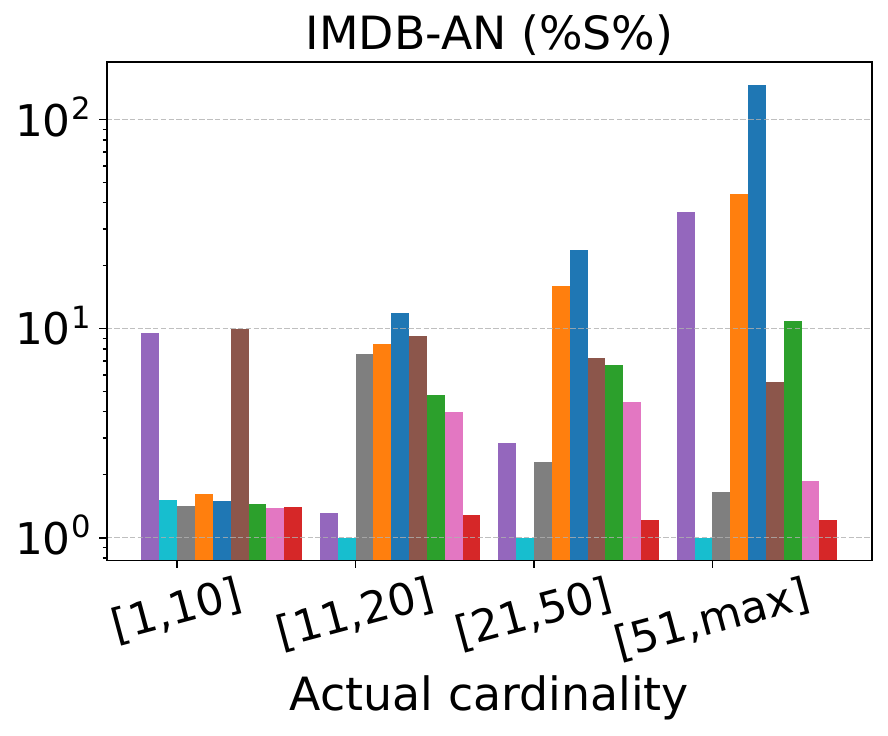}
    \label{fig:imdb_an_diff_range_sub}
  \end{subfigure}%
  \hfill
  \begin{subfigure}[b]{0.24\textwidth}
    \centering
    \includegraphics[width=\textwidth]{ 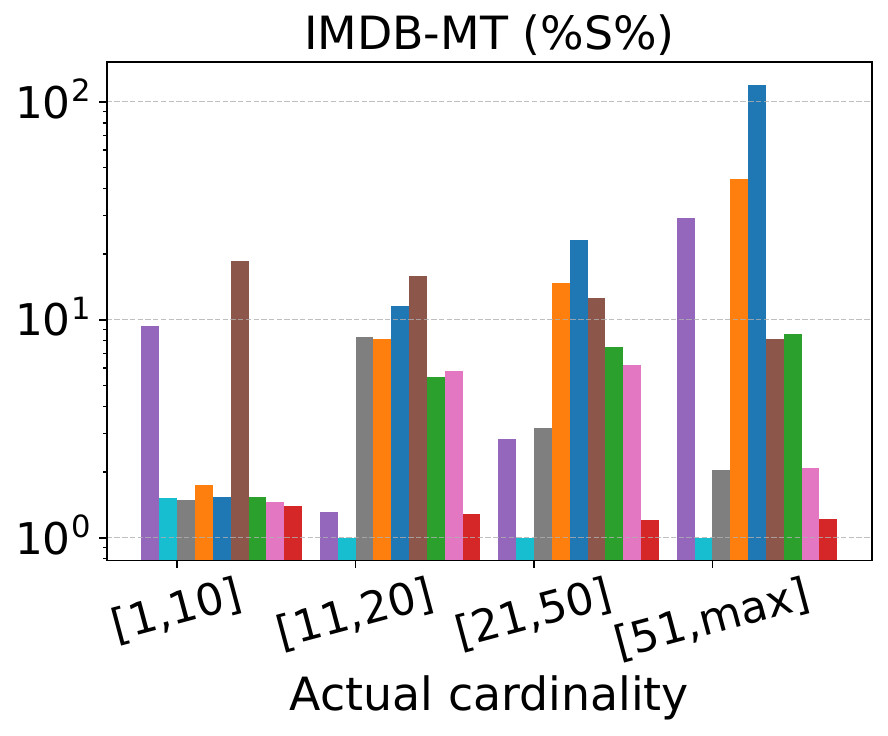}
    \label{fig:imdb_mt_diff_range_sub}
  \end{subfigure}

    \vspace{-7mm}
  \caption{Q-error of all estimators varying actual result cardinality.}
  \label{fig:diff_range}
  \vspace{-4mm}
\end{figure*}
\noindent\textbf{Robustness on Different Cardinality Ranges}. Figure~\ref{fig:diff_range} reports mean Q-error of each estimator based on actual cardinality range across all datasets and query patterns. We observe the following: (1) Although most LIKE queries have small cardinalities, a non-negligible fraction exhibit moderate or large cardinalities. Existing methods perform poorly on these queries, often incurring extremely large Q-errors that can severely distort optimizer decisions. (2) Our method almost maintains the lowest Q-error across all cardinality ranges and datasets, showing strong robustness. It achieves at least 2x lower Q-error that others (except for MO and PostgreSQL ) when the true cardinality is larger than 10. (3) PostgreSQL (PG) achieves a low Q-error in the range $[11,20]$. This is because it usually predicts the cardinality with $11$ in our testing. (4) Learning-based methods, e.g., E2E and CLIQUE, tend to exhibit increasing Q-error as the true cardinality grows, indicating degraded estimation accuracy on large-result queries. These models struggle to generalize to high-cardinality patterns, likely due to skewed training distributions and overfitting to small range cases.

\vspace{-4mm}
\subsection{Estimation of Empty-answer Queries}\label{exp:emtpy}
We report the percentages of queries with estimated cardinality smaller than $2$. For \BF{}, we compare: \BF{}-\textsf{N} -- the naive method in Sec.~\ref{sec:neg_naive}, and \BF{}-\textsf{PW} -- prefix-walk strategy in Sec.~\ref{sec:enhence-neg}. From Table~\ref{tab:neg-acc}, we can see: (1) In all datasets, \BF{}-\textsf{PW} significantly outperforms \BF{}-\textsf{N}, which aligns with our theoretic analysis in Sec.~\ref{sec:enhence-neg}. 
(2) LEARNT-PW is the best in $S\%$ and $\%S$ while is competitive with LBS, MO, and SSCard
in $\%S\%$.
(3) Only \BF{}-\textsf{PW} has consistently high performance across datasets. In contrast, other methods exhibit fluctuations. For example, CLIQUE achieves 95.09 on Author while only 84.20 on DBLP-AN.

\begin{table}[ht]
\centering
\vspace{-2mm}
\caption{Probability (Prob.) of estimates below 2}
\vspace{-3mm}
\resizebox{0.50\textwidth}{!}{
\begin{tabular}{|c|c|c|c|c|c|}
\hline
{\textbf{Query}} & {\textbf{Method}} 
& {\textbf{Author}} 
& {\textbf{DBLP-AN}} 
& {\textbf{IMDB-AN}}
& {\textbf{IMDB-MT}} \\\hline
\multirow{6}{*}{ S\%} 
& PostgreSQL & 0 & 0  & 0 & 0  \\
& \BF{}-\textsf{N}     & 93.98  & 88.37  & 97.47 & 95.90 \\
& \BF{}-\textsf{PW}     & 99.47  & 98.55  & 99.86 & 99.74 \\\cline{2-6}
& E2E         &   4.87    &    4.19     &  26.99    &   0.02          \\
& Astrid      &    92.16   &   78.30     & 97.99  & 95.36  \\
& LPLM        & 22.98  & 22.76  & 31.78 & 53.67 \\
& CLIQUE      &  96.09    &   92.48      &   99.50   &    99.44       \\
\hline\hline
\multirow{6}{*}{ \%S} 
& PostgreSQL & 0 & 0  & 0 & 0  \\
& \BF{}-\textsf{N}     & 93.48  & 88.35  & 97.15 & 96.23  \\
& \BF{}-\textsf{PW}     & 99.44  & 98.58  & 99.82 & 99.76 \\\cline{2-6}
& E2E         &   3.12    &    2.32     &  29.12    &   0.02  \\
& Astrid      &   90.33    &  76.37      & 99.29 &  98.08  \\
& LPLM        & 24.01  & 11.04  & 17.21 & 43.21  \\
& CLIQUE      &  95.09     &   84.20       &   98.83  &  99.35           \\

\hline\hline
\multirow{10}{*}{ \%S\% } 
& PostgreSQL & 0 & 0  & 0 & 0  \\
& MO & 99.21 & 99.43  & 99.78 & 99.80 \\
& LBS & 99.99 & 99.97  & 99.97 & 99.97 \\
& \BF{}-\textsf{N}     & 90.31  & 85.40  &95.62 & 95.17\\
& \BF{}-\textsf{PW}     & 99.10  & 98.15  &99.73 & 99.68 \\\cline{2-6}
& E2E         &   15.04    &   6.43      &    27.69          &   0.02        \\
& Astrid      &   89.16    &    86.42     & 99.43  & 98.69   \\
& LPLM        & 11.53  & 6.23  & 12.66 & 20.88 \\
& CLIQUE      &  95.44     &   92.82       &   99.42   &    99.49          \\
& SSCard & 99.92 & 99.88  & 99.75 & 99.90  \\
\hline
\end{tabular}\label{tab:neg-acc}
}
\vspace{-1mm}
\end{table}

\vspace{-3mm}
\subsection{Overhead}
\noindent\textbf{Preprocessing Time}. Table~\ref{tab:prep_time} shows preprocessing time of each estimator before supporting online inference. \BF{} achieves the lowest offline construction time across all datasets by a large margin. Compared to E2E, which takes up to around 43,000 seconds, \BF{} completes in just 17–77 seconds, yielding over 500x speedup, demonstrating that \BF{} is highly efficient for offline preprocessing and well-suited for practical deployment.

\begin{table}[ht]
\centering
\vspace{-1mm}
\caption{Preprocessing time (s) comparison across datasets}
\vspace{-2mm}
\begin{tabular}{|c|c|c|c|c|}
\hline
\textbf{Method} & \textbf{Author} & \textbf{DBLP-AN} & \textbf{IMDB-AN} & \textbf{IMDB-MT} \\\hline
MO   & 40.27  & 133.98  & 267.26  & 486.56 \\\hline

LBS   & 219.92 & 923.92  & 997.13  & 974.31 \\\hline
\BF{}    & 17.13  & 44.11  & 76.78  & 71.33 \\\hline\hline
E2E     & 26970  & 36896  & 43791  & 41080 \\\hline
Astrid  &   22519     &    38143    &  56607      &   51696    \\\hline
LPLM    & 9419   & 49767   & 65376   & 45457  \\\hline
CLIQUE  & 6704   & 9935   & 10879  & 11815 \\\hline
SSCard   & 27.19  & 141.49  & 117.31  & 102.53 \\\hline
\end{tabular}\label{tab:prep_time}
\end{table}

\noindent\textbf{Inference Time}. 
As shown in Table~\ref{tab:infer_time}, \BF{} achieves low inference latency, significantly outperforming most learning-based approaches. Specifically, it is over 80× faster than E2E, and substantially faster than Astrid, LPLM, and CLIQUE across all datasets. Compared to non-learning-based methods such as MO and LBS, \BF{} incurs only a slight overhead (e.g., 0.05ms vs. 0.01–0.12ms), while still maintaining competitive latency.

\noindent\textbf{Storage Usage}. As shown in Table~\ref{tab:storage}, \BF{} achieves the lowest storage usage among all non-learning-based methods, demonstrating strong efficiency advantages over traditional techniques like MO and LBS. Compared to learning-based methods (Astrid, E2E, CLIQUE, LPLM, and SSCard), \BF{} uses less storage than three of them (E2E, SSCard, and Astrid) and is competitive with the most compact ones (CLIQUE and LPLM). 

\subsection{Effectiveness of Long Queries Support}
\tocheck{Table~\ref{tab:long_query} shows mean Q-error of \BF{} and estimation time on the queries with the length between 11 and 20. We can observe that \BF{} achieves consistently low Q-error (around 1.3–1.4) on long queries across all datasets and query types, while keeping estimation latency below 1 ms. This shows that the Markov-based extension provides accurate and efficient support.}

\begin{table}[ht]
\vspace{-2mm}
\centering
\caption{Inference time (ms) comparison across datasets}
\vspace{-2mm}
\begin{tabular}{|c|c|c|c|c|}
\hline
\textbf{Method} & \textbf{Author} & \textbf{DBLP-AN} & \textbf{IMDB-AN} & \textbf{IMDB-MT} \\
\hline
PostgreSQL & 0.09  & 0.09  & 0.09  & 0.09 \\\hline
MO & 0.02  & 0.01  & 0.02  & 0.02 \\\hline
LBS & 0.09 & 0.12  & 0.12  & 0.11 \\\hline
\BF{}       & 0.04  & 0.05  & 0.04  & 0.04 \\\hline\hline
E2E        & 6.90  & 6.90  & 6.90  & 7.10 \\\hline
Astrid     &  3.27     &  3.18     &  3.26     &   3.27   \\\hline
LPLM       & 1.30  & 1.00  & 1.00  & 1.30 \\\hline
CLIQUE     & 3.65  & 3.90  & 4.00  & 3.70 \\\hline
SSCard & 0.02  & 0.03  & 0.03  & 0.03 \\\hline
\end{tabular}\label{tab:infer_time}
\end{table}

\begin{table}[ht]
\vspace{-6mm}
\centering
\caption{Storage (MB) comparison across datasets}
\vspace{-2mm}
\begin{tabular}{|c|c|c|c|c|}
\hline
\textbf{Method} & \textbf{Author} & \textbf{DBLP-AN} & \textbf{IMDB-AN} & \textbf{IMDB-MT} \\
\hline
MO    & 2.45  & 7.76  & 13.80  & 24.64 \\\hline
LBS    & 8.64  & 15.73  & 15.68  & 16.42 \\\hline
\BF{}    & 0.51  & 1.77  & 1.73  & 1.76 \\\hline\hline
E2E     & 20.84 & 6.52  & 6.41  & 6.37 \\\hline
Astrid  &   2.64    &  2.04     &  2.04     &  2.04    \\\hline
LPLM    & 1.09  & 1.02  & 1.02  & 1.05 \\\hline
CLIQUE  & 2.21  & 1.04  & 1.03  & 1.06 \\\hline
SSCard    & 1.65  & 4.39  & 4.29  & 4.29 \\\hline
\end{tabular}\label{tab:storage}
\vspace{-4mm}
\end{table}






\begin{table}[ht]
\centering
\vspace{-2mm}
\caption{Q-error and estimation time (ms) on long queries}
\label{tab:long_query}
\vspace{-2mm}
\begin{tabular}{|c|c|c|c|c|c|c|c|c|}
\hline
\textbf{Query} 
& \multicolumn{2}{c|}{\textbf{Author}} 
& \multicolumn{2}{c|}{\textbf{DBLP-AN}} 
& \multicolumn{2}{c|}{\textbf{IMDB-AN}} 
& \multicolumn{2}{c|}{\textbf{IMDB-MT}} \\\hline
$S\%$   & 1.40 & 0.34 & 1.42 & 0.42 & 1.33 & 0.31 & 1.28 & 0.47 \\\hline
$\%S$   & 1.39 & 0.31 & 1.42 & 0.38 & 1.30 & 0.37 & 1.27 & 0.46 \\\hline
$\%S\%$ & 1.39 & 0.30 & 1.42 & 0.38 & 1.31 & 0.32 & 1.28 & 0.43 \\\hline
\end{tabular}
\vspace{-6mm}
\end{table}

\subsection{Impact on Query Optimization}
\noindent{\revision{\textbf{Settings.} We evaluate 57 JOB queries containing prefix, suffix, and substring \texttt{LIKE} predicates. For fairness, estimators are constructed 
for all columns involving \texttt{LIKE}. We compare 
PostgreSQL, CLIQUE~\cite{CLIQUE}, SSCard~\cite{SSCard}, and \BF{} on PostgreSQL v13.1. 
Parallel execution is disabled and a 1000s timeout is enforced. 
Each query is executed three times per method.}}

\begin{table}[!ht]
\centering
\small
\vspace{-2mm}
\caption{\revision{Execution time comparison on JOB benchmark}}
\vspace{-2mm}
\label{tab:execution}
\begin{tabular}{|l|cc|cc|c|}
\hline
\multirow{2}{*}{\textbf{Method}} 
& \multicolumn{2}{c|}{\textbf{Improved}} 
& \multicolumn{2}{c|}{\textbf{Degradation}} 
& \multirow{2}{*}{\textbf{\shortstack{Total \\ Runtime (s)}}} \\
\cline{2-5}
& \#Q & Avg. Ratio (\%) 
& \#Q & Avg. Ratio (\%) 
&  \\
\hline
PostgreSQL      & -- & --   & -- & --    & 1205.79 \\\hline
SSCard  & 8  & 52.6 & 49 & 17.4  & 1182.17 \\\hline
CLIQUE  & 9  & 36.3 & 48 & 2254.6 & 3192.38 \\\hline
\BF{}  & 30 & 16.5 & 27 & 4.0   & 1106.99 \\
\hline
\end{tabular}
\vspace{-3mm}
\end{table}

\noindent\revision{\textbf{Result.} Table~\ref{tab:execution} reports the number of queries with performance improvement/degradation compared to PostgreSQL, the average improvement and degradation ratios, and the total runtime. We observe: (1) {\BF{} achieves the lowest total runtime}, showing that its accuracy improvements translate into execution gains. (2) \BF{} improves the largest number of queries, indicating 
more consistent plan improvements. (3) Although \BF{} has 27 degradations, the {average slowdown is only 4\%}, much smaller than others whose severe degradation offset its improvements and inflate total runtime.}

\vspace{-2mm}
\subsection{Effectiveness of Our Design}

\noindent\textbf{Performance on Different $eb$ Settings}. Table~\ref{tab:eb-eval} summarizes the trade-offs of \BF{} under different $eb$s.  We test on DBLP-AN with $\%S\%$ ad other datasets and query patterns have similar trend. As $eb$ increases from 1.3 to 2.5, the storage usage, build time, and inference time significantly decrease (e.g., storage drops from 2.09MB to 0.53MB), indicating better space and runtime efficiency. However, this comes at the cost of higher mean Q-error, which rises from 1.09 to 2.01. This highlights that \BF{} provides a tunable balance between accuracy and efficiency, with smaller $eb$ favoring estimation quality, and larger $eb$ reducing resource consumption.
\begin{table}[!ht]
	\small
	\vspace{-2mm}
	\centering
	\caption{Analysis of \BF{} with different $eb$s.}
	\vspace{-4mm} 
	\resizebox{0.48\textwidth}{!}{%
		\begin{tabular}{|c|c|c|c|c|c|}
			\hline
			\textbf{$eb$} & 1.3 & 1.5 & 1.7 & 2.0 & 2.5 \\\hline
                Building time (s) &42.09&34.15&34.78&31.69&27.76\\\hline
                Storage (MB)       &2.09&1.36&1.27&0.76&0.53\\\hline
                mean Q-error  &1.09&1.37&1.4&1.71&2.01\\\hline
                infer. time (ms)  &0.12&0.07&0.07&0.04&0.03\\\hline
        \end{tabular}%
	}
	\label{tab:eb-eval}
    \vspace{-2mm}
\end{table}

\noindent\textbf{Performance on Different $p_n$ Settings and Impact of Upper Tree}. Table~\ref{tab:upper-tree} shows the impact of enabling the upper tree and tuning $p_n$ ( the probability of
classifying empty-answer query with $B_1$) on storage and empirical probability. 'Unset' indicates that no constraint is imposed on empty-answer query classification probability. We test on DBLP-AN with
$\%S\%$ and other datasets and query patterns have similar trend. Across all $p_n$ settings, using the upper tree consistently reduces storage usage while maintaining comparable or slightly better accuracy.

\begin{table}[ht]
\vspace{-2mm}
\centering
\caption{Effect of upper tree and $p_n$ thresholds}
\vspace{-2mm}
\resizebox{0.48\textwidth}{!}{%
\begin{tabular}{|c|c|c|c|c|c|c|c|c|}
\hline
 \multirow{2}{*}{\textbf{Variants}}& \multicolumn{2}{c|}{\textbf{Unset}} & \multicolumn{2}{c|}{\textbf{99.9\%}} & \multicolumn{2}{c|}{\textbf{99.99\%}} & \multicolumn{2}{c|}{\textbf{99.999\%}} \\\cline{2-9}
 & Stor. & Prob. & Stor. & Prob. & Stor. & Prob. & Stor. & Prob. \\
\hline
\textbf{w/o Tree} & 1.36 & 96.11\% & 2.43 & 99.9\% & 2.88 & 99.97\% & 3.34 & 99.993\% \\\hline
\textbf{w/ Tree}  & 1.36 & 98.15\% & 1.92 & 99.8\% & 2.33 & 99.93\% & 2.77 & 99.979\% \\
\hline
\end{tabular}\label{tab:upper-tree}
}
\vspace{-2mm}
\end{table}

\noindent\textbf{Impact of Frontier-based Pruning}.
Table~\ref{tab:ft-opt} shows that applying frontier-based optimization significantly reduces both storage and build time across all datasets. On average, storage drops by 20–30\% (e.g., from 2.00MB to 1.48MB on IMDB-MT), and build time is reduced by up to 40\% (e.g., from 103.07s to 61.01s on IMDB-MT). 

\begin{table}[ht]
\vspace{-2mm}
\centering
\caption{Impact of frontier-based pruning}
\vspace{-2mm}
\resizebox{0.48\textwidth}{!}{%
\begin{tabular}{|c|c|c|c|c|c|c|c|c|}
\hline
 \multirow{2}{*}{\textbf{Variants}}& \multicolumn{2}{c|}{\textbf{Author}} & \multicolumn{2}{c|}{\textbf{DBLP-AN}} & \multicolumn{2}{c|}{\textbf{IMDB-AN}} & \multicolumn{2}{c|}{\textbf{IMDB-MT}} \\\cline{2-9}
 & Stor. & Time & Stor. & Time & Stor. & Time & Stor. & Time \\
\hline
\textbf{w/o F} & 0.54 & 21.49 & 1.80 & 52.58 & 1.79 & 97.68 & 2.00 & 103.07 \\\hline
\textbf{w/ F}  & 0.40 & 13.16 & 1.36 & 34.15 & 1.33 & 60.23 & 1.48 & 61.01 \\
\hline
\end{tabular}\label{tab:ft-opt}
}
\vspace{-2mm}
\end{table}

\section{Conclusion}
In this paper, we study the problem of cardinality estimation for LIKE queries. We introduce \BF{}, a novel classification-based estimator that transforms cardinality estimation into a tunable, bucket-based classification problem. \BF{} combines a bucketed layered filter architecture with theoretical guarantees for valid queries and effective support for negative queries. It further incorporates a parameter selection framework that minimizes storage under accuracy constraints.
Our extensive experiments across real-world datasets show that \BF{} achieves consistently lower mean and tail Q-errors than state-of-the-art methods such as CLIQUE and LPLM, while also reducing preparation time by up to 70×.



\bibliographystyle{ACM-Reference-Format}
\bibliography{sample}

\end{document}